\documentclass[hidelinks,onefignum,onetabnum]{siamart251216}

\usepackage{amsfonts,amssymb,amsmath,latexsym}
\usepackage{enumitem}
\usepackage{tikz}
\usepackage{tikz-cd}
\usetikzlibrary{arrows.meta,positioning,fit,calc}
\usepackage{pgfplots}
\pgfplotsset{compat=1.17}
\usepackage{algorithm}
\usepackage{algorithmic}
\usepackage{makecell}
\definecolor{bordercol}{HTML}{222222}
\definecolor{bgL1}{HTML}{F4F7F9}
\definecolor{bgL2}{HTML}{E2EAF1}
\definecolor{bgL3}{HTML}{C9D8E6}
\definecolor{bgL4}{HTML}{AEC5D9}
\definecolor{bgL5}{HTML}{8CAECA}
\tikzset{
	basebox/.style={
		rectangle, rounded corners=1.5pt,
		draw=bordercol, line width=0.6pt,
		align=center, font=\sffamily\small\color{bordercol},
		inner sep=0.15cm, outer sep=0pt
	},
	L1box/.style={basebox, fill=bgL1, minimum width=6.8cm, minimum height=1.2cm, text width=6.4cm},
	L2box/.style={basebox, fill=bgL2, minimum width=6.8cm, minimum height=1.2cm, text width=6.4cm},
	L3box/.style={basebox, fill=bgL3, minimum width=14.0cm, minimum height=1.1cm, text width=13.6cm, font=\sffamily\small\bfseries\color{bordercol}},
	L4box/.style={basebox, fill=bgL4, minimum width=3.2cm, minimum height=1.4cm, text width=2.8cm},
	L5box/.style={basebox, fill=bgL5, minimum width=14.0cm, minimum height=1.1cm, text width=13.6cm, font=\sffamily\small\bfseries\color{bordercol}}
}

\ifpdf
\hypersetup{
	pdftitle={Rank Distribution and Dynamics of Gram Matrices from Binary m-Sequences with Applications to LCD Codes},
	pdfauthor={Authors}
}
\fi

\makeatletter
\@ifundefined{theorem}{\newtheorem{theorem}{Theorem}[section]}{}
\@ifundefined{lemma}{\newtheorem{lemma}[theorem]{Lemma}}{}
\@ifundefined{proposition}{\newtheorem{proposition}[theorem]{Proposition}}{}
\@ifundefined{corollary}{}{}
\makeatother

\theoremstyle{definition}
\makeatletter
\@ifundefined{definition}{\newtheorem{definition}[theorem]{Definition}}{}
\@ifundefined{example}{\newtheorem{example}[theorem]{Example}}{}
\@ifundefined{notation}{}{}
\makeatother

\newtheorem{open}[theorem]{Open Problem}

\theoremstyle{remark}
\makeatletter
\@ifundefined{remark}{}{}
\makeatother

\providecommand{\ord}{{\mathrm{ord}}}

\providecommand{\rank}{{\mathrm{rank}}}

\providecommand{\Tr}{{\mathrm{Tr}}}

\providecommand{\F}{{\mathbb{F}}}

\definecolor{matlabblue}{rgb}{0, 0.4470, 0.7410}

\begin{document}
	
	\title{Rank Distribution and Dynamics of Gram Matrices from Binary m-Sequences with Applications to LCD Codes}
	
	\author{%
		Hengfeng Liu\thanks{School of Mathematics, Southwest Jiaotong University, Chengdu, China (\email{hengfengliu@163.com}).}%
		\and
		Chunming Tang\thanks{School of Information Science and Technology, Southwest Jiaotong University, Chengdu, China (\email{tangchunmingmath@163.com}).}%
        \and
	Cuiling Fan\thanks{School of Mathematics, Southwest Jiaotong University, Chengdu, China (\email{cuilingfan@163.com}).}
    \and
    Zhengchun Zhou\thanks{School of Information Science and Technology, Southwest Jiaotong University, Chengdu, China (\email{zzc@swjtu.edu.cn}).}}

	\maketitle
	
	\begin{abstract}
The Gram matrix is a classical object formed from the pairwise inner products of a collection of vectors, with fundamental roles in functional analysis, statistics, combinatorics, and coding theory. In the realm of sequence design, maximum-length sequences (m-sequences) are among the most fundamental classes of sequences, traditionally characterized by their span, decimation, shift-and-add, balance, run, and ideal autocorrelation properties.

In this paper, we bridge the two foundational concepts by uncovering novel structural features of m-sequences through the lens of a family of Gram matrices. Specifically, for each $1 \le t \le 2^n - 1$, we extract $n$ consecutive subsequences of length $t$ from an m-sequence of period $2^n - 1$, construct their corresponding $n \times n$ Gram matrix, and investigate its rank, denoted by $r_n(t)$. Utilizing semilinear representation of Galois groups and B\'ezoutian of polynomials, we derive an explicit formula for $r_n(t)$ for all $t$, thereby establishing the complete rank distribution of these Gram matrices. Notably, we prove that full rank is attained for approximately half of the admissible values of $t$. We further uncover the intricate dynamics of $r_n(t)$: rank-deficient states are strictly unstable (i.e., $r_n(t) < n$ implies $r_n(t+1) \ne r_n(t)$), whereas the full-rank state exhibits strong persistence, remaining at $n$ over a nontrivial interval of consecutive values of $t$. Altogether, our results fully characterize both the global rank distribution and the local dynamics of rank function, as invariant of m-sequences. As an application, our findings completely determine the hull distribution of the family of punctured cyclic simplex codes.
\end{abstract}
	
	

\section{Introduction}

Binary m-sequences form a distinguished class of periodic sequences generated by linear feedback shift registers (LFSRs). Their systematic study goes back to the pioneering work of Solomon Golomb in the 1960s, whose book \textit{Shift Register Sequences} \cite{Golomb 1967} established much of the basic theory of shift register sequences.

Let $\mathbb{F}_2$ be the binary field. An $n$-stage linear feedback shift register (LFSR) generates a binary periodic sequence $\{s_t\}_{t=0}^{\infty}$ satisfying
\begin{equation}
	s_{t+n}=c_{n-1}s_{t+n-1}+\dots+c_1s_{t+1}+c_0s_t,
\end{equation}
where $c_0,\dots,c_{n-1}\in \mathbb{F}_2$. Such an LFSR sequence is determined by its initial state $(s_0,s_1,\dots,s_{n-1})\in \mathbb{F}_2^n$ and its characteristic polynomial
\[
f(x)=x^n+c_{n-1}x^{n-1}+\dots+c_1x+c_0\in\mathbb{F}_2[x].
\]

When $f(x)$ is a primitive polynomial of degree $n$, the resulting sequence has the maximum possible period $2^n-1$ and is called a binary \textit{maximum-length sequence}, or simply a binary \textit{m-sequence} of order $n$. The following figure illustrates an LFSR that generates an m-sequence. 

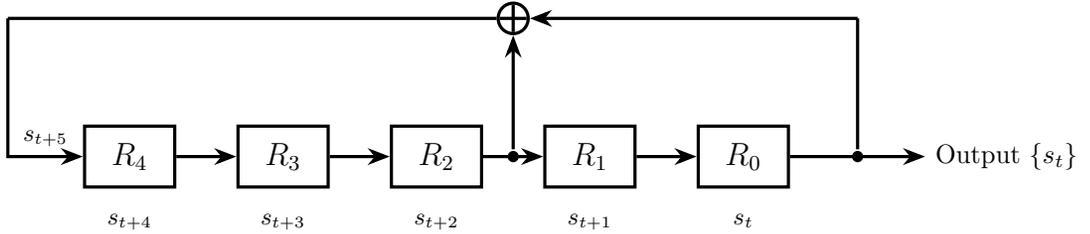
\begin{figure}[htbp]
	\centering
	\begin{tikzpicture}[
		>=Stealth,
		very thick, 
		reg/.style={draw, rectangle, minimum width=1.2cm, minimum height=0.8cm, font=\large},
		xor/.style={draw, circle, inner sep=0pt, minimum size=0.4cm, append after command={
				[shorten >=\pgflinewidth, shorten <=\pgflinewidth,]
				(\tikzlastnode.north) edge (\tikzlastnode.south)
				(\tikzlastnode.east) edge (\tikzlastnode.west)
		}},
		tap/.style={circle, fill, inner sep=1.5pt} 
		]
		
		\node[reg] (r4) {$R_4$};
		\node[reg, right=0.8cm of r4] (r3) {$R_3$};
		\node[reg, right=0.8cm of r3] (r2) {$R_2$};
		\node[reg, right=0.8cm of r2] (r1) {$R_1$};
		\node[reg, right=0.8cm of r1] (r0) {$R_0$};
		
		\node[below=0.2cm of r4, font=\small] {$s_{t+4}$};
		\node[below=0.2cm of r3, font=\small] {$s_{t+3}$};
		\node[below=0.2cm of r2, font=\small] {$s_{t+2}$};
		\node[below=0.2cm of r1, font=\small] {$s_{t+1}$};
		\node[below=0.2cm of r0, font=\small] {$s_t$};
		
		\draw[->] (r4) -- (r3);
		\draw[->] (r3) -- (r2);
		\draw[->] (r2) -- (r1) node[midway, tap] (t2) {};
		\draw[->] (r1) -- (r0);
		
		\draw[->] (r0.east) -- ++(1.8,0) node[right] {Output $\{s_t\}$} node[midway, tap] (t0) {};
		
		\node[xor, above=1.5cm of t2] (x1) {};
		
		\draw[-] (t0) -- (t0 |- x1.east) coordinate (c1);
		\draw[->] (c1) -- (x1.east);
		\draw[->] (t2) -- (x1.south);
		
		\coordinate (feedback_left) at ([xshift=-1cm]r4.west);
		\draw[-] (x1.west) -- (x1.west -| feedback_left) coordinate (c2);
		\draw[->] (c2) |- (r4.west) node[near end, above, font=\small] {$s_{t+5}$};
		
	\end{tikzpicture}
	\caption{A $5$-stage linear feedback shift register with $f(x) = x^5 + x^2 + 1$.}
	\label{fig:lfsr}
\end{figure}

Over the past five decades, m-sequences have served as indispensable building blocks in digital communications, radar synchronization, and modern cryptography \cite{Helleseth Li 2024}. Beyond these applications, their rich algebraic and combinatorial structures are closely related to some fundamental topics in other disciplines, including cryptography \cite{Carlet 2020,Ding Xiao Shan 1991,Xia 2017,Xiong 2024}, coding theory \cite{Ding 2013,Ding 2018 sequence construction,Li 2014,Xiong 2015}, number theory \cite{Helleseth 1976}, finite geometry \cite{Games 1986}, and combinatorial \cite{Xiang 2024}. 
Golomb \cite{Golomb 1967} characterized several remarkable properties of m-sequences, which we list below.
 
\begin{enumerate}
	\item \textit{Span Property:} Each nonzero $n$-tuple over $\mathbb{F}_2$ occurs exactly once in one period of $\{s_t\}$.
	
	\item \textit{Decimation Property:} For any positive integer $d$ coprime to $2^n-1$, the decimated sequence $\{s_{dt}\}$ is also an m-sequence of period $2^n-1$.
	
	\item \textit{Shift-and-Add Property:} For any shift $\tau \not\equiv 0 \pmod{2^n-1}$, the sequence $\{s_t'\}$ with $s_t' = s_{t+\tau} + s_t$ is also an m-sequence of period $2^n-1$.
	
	\item \textit{Balance Property:} In one period of $\{s_t\}$, the element $1$ occurs exactly $2^{n-1}$ times, and the zero element occurs $2^{n-1}-1$ times.
	
	\item \textit{Run Property:} In one period of $\{s_t\}$, exactly $1/2$ of the runs have length $1$, $1/4$ have length $2$, $1/8$ have length $3$, $1/16$ have length $4$, and so on, as long as these fractions give integral numbers of runs. Here, a run of length $k$ is defined as a block $s_i \dots s_{i+k-1}$ such that $s_{i-1} \neq s_i = \dots = s_{i+k-1} \neq s_{i+k}$.
	
	\item \textit{Ideal Autocorrelation Property:} The periodic autocorrelation function $C_s(\tau)$ is given by
	\begin{equation*}
		C_s(\tau) = \sum_{t=0}^{2^n-2} (-1)^{s_{t+\tau} - s_t} = 
		\begin{cases} 
			2^n - 1, & \text{if } \tau \equiv 0 \pmod{2^n-1}, \\ 
			-1, & \text{if } \tau \not\equiv 0 \pmod{2^n-1}. 
		\end{cases}
	\end{equation*}
\end{enumerate}

In particular, Golomb recognized the importance of the balance, run, and ideal autocorrelation properties in the context of randomness. Binary sequences satisfying these three properties are called \textit{pseudo-noise (PN) sequences}. Due to their good pseudo-random properties, PN sequences have many practical applications in digital communications, including navigation, radar, and spread-spectrum systems \cite{Helleseth Li 2021}. Other interesting properties of m-sequences can be found in \cite{Golomb 1967,Helleseth 1976,Trachtenberg 1970}.

Let $\{s_t\}_{t=0}^\infty$ be a binary m-sequence of order $n$. For each integer $t$ with $1 \le t \le 2^n-1$, we take the $n$ consecutive subsequences of length $t$ beginning at positions $0,1,\dots,n-1$, and arrange them into the matrix
\begin{equation}
G_t=
\begin{pmatrix}
	s_0 & s_1 & \cdots & s_{t-1}\\
	s_1 & s_2 & \cdots & s_t\\
	\vdots & \vdots & \ddots & \vdots\\
	s_{n-1} & s_n & \cdots & s_{n+t-2}
\end{pmatrix}
\in \mathbb{F}_2^{\,n\times t}.
\end{equation}
 $G_t$ is the \textit{observability matrix} associated with the m-sequence. Here, each column
\[
(s_j,s_{j+1},\dots,s_{j+n-1})^\top \in \mathbb{F}_2^{\,n\times 1}
\]
is exactly the state vector of the underlying $n$-stage LFSR at time $j$, and thus $G_t$ may also be viewed as the matrix formed by the consecutive state vectors from time $0$ to time $t-1$. It is also easy to see that $\rank(G_{t})=\min\{n,t\}$.

Given a matrix $G$, its Gram matrix $GG^{\top}$ is formed from the pairwise inner 
products of row vectors. Since the Gram matrix is a classical object, it is natural to consider the Gram matrix of the observability matrix $G_{t}$. In particular, one can ask the following rank questions of $G_tG_t^{\top}$, and the rank evolution and distribution of the cases $n=5,6$ are illustrated in Figure \ref{fig:gram_rank_all}. 

\vspace{0.5em}
\noindent{\emph{Question 1:}} Is there a general characterization for $\rank(G_t G_t^\top)$?
\vspace{0.5em}

\noindent{\emph{Question 2:}}  What can we say about distribution $\displaystyle \#\{\, 1 \le t \le 2^n - 1 \mid\rank(G_t G_t^\top) = k \,\}$, \ for $0 \le k \le n$? In particular, how many of the Gram matrices are full-rank?

 \vspace{0.3em}

In this paper, we not only answer the two questions, but also reveal new properties of m-sequences that go beyond the classical ones established by Golomb. More precisely, we derive new properties of $\rank(G_t G_t^\top)$, which are independent of the choice of primitive polynomials attached to m-sequences. Our first goal is to determine the rank distribution of the family $\{G_t G_t^\top \mid 1 \le t \le 2^n-1\}$. To this end, we first characterize the singular matrices of the family $\{G_tG_t^\top \mid 1\le t\le 2^n-1\}$ by a canonical representation of rational functions. This representation leads to an explicit rank characterization for singular matrices, and thereby reduces the determination of the whole rank distribution to the enumeration of some special rational functions. Next, we also investigate the rank dynamics along the natural cyclic ordering, and some interesting properties are presented. Finally, some applications to coding theory are reported.

\begin{figure}[htbp]
	\centering
	\begin{minipage}[t]{0.49\textwidth}
		\centering
		\begin{tikzpicture}
			\begin{axis}[
				width=\textwidth,
				height=6cm,
				xlabel={$t$},
				ylabel={$\mathrm{rank}(G_t G_t^\top)$},
				xmin=1, xmax=31.2,
				ymin=0, ymax=5.5,
				xtick={0,5,10,15,20,25,30},
				extra x ticks={1,31},
				extra x tick labels={$1$,},
				extra x tick style={
					major tick length=1.5pt,
					grid=none,
					tick label style={yshift=-0.5ex}
				},
				extra x tick style={major tick length=1.5pt, grid=none},
				minor x tick num=4,
				ytick={0,1,2,3,4,5},
				grid=both,
				major grid style={dashed, gray!60},
				minor grid style={dotted, gray!30},
				tick align=outside,
				tick pos=left,
				axis lines*=left
				]
				\addplot[
				color=blue,
				thick,
				mark=*,
				mark options={fill=red, draw=red, scale=0.5}
				] coordinates {
					(1,1) (2,2) (3,3) (4,4) (5,5) (6,5) (7,5) (8,4) (9,5) (10,5)
					(11,4) (12,3) (13,4) (14,5) (15,5) (16,5) (17,5) (18,4) (19,3) (20,4)
					(21,5) (22,5) (23,4) (24,5) (25,5) (26,5) (27,4) (28,3) (29,2) (30,1)
					(31,0)
				};
			\end{axis}
		\end{tikzpicture}
		\par\vspace{1mm}
		\makebox[\textwidth][c]{\small (a) $n=5$, $f(x)=x^5+x^2+1$.}
	\end{minipage}\hfill
	\begin{minipage}[t]{0.49\textwidth}
		\centering
		\begin{tikzpicture}
			\begin{axis}[
				width=\textwidth,
				height=6cm,
				xlabel={$t$},
				ylabel={$\mathrm{rank}(G_t G_t^\top)$},
				xmin=1, xmax=31.2,
				ymin=0, ymax=5.5,
				xtick={0,5,10,15,20,25,30},
				extra x ticks={1,31},
				extra x tick labels={$1$,},
				extra x tick style={
					major tick length=1.5pt,
					grid=none,
					tick label style={yshift=-0.5ex}
				},
				extra x tick style={major tick length=1.5pt, grid=none},
				minor x tick num=4,
				ytick={0,1,2,3,4,5},
				grid=both,
				major grid style={dashed, gray!60},
				minor grid style={dotted, gray!30},
				tick align=outside,
				tick pos=left,
				axis lines*=left
				]
				\addplot[
				color=blue,
				thick,
				mark=*,
				mark options={fill=red, draw=red, scale=0.5}
				] coordinates {
					(1,1) (2,2) (3,3) (4,4) (5,5) (6,5) (7,4) (8,3) (9,4) (10,5)
					(11,5) (12,5) (13,5) (14,4) (15,5) (16,5) (17,4) (18,5) (19,5) (20,5)
					(21,5) (22,4) (23,3) (24,4) (25,5) (26,5) (27,4) (28,3) (29,2) (30,1)
					(31,0)
				};
			\end{axis}
		\end{tikzpicture}
		\par\vspace{1mm}
		\makebox[\textwidth][c]{\small (b) $n=5$, $f(x)=x^5+x^3+x^2+x+1$.}
	\end{minipage}
	
	\vspace{6mm}
	
	\begin{minipage}[t]{0.49\textwidth}
		\centering
		\begin{tikzpicture}
			\begin{axis}[
				width=\textwidth, 
				height=6cm,
				xlabel={$t$},
				ylabel={$\mathrm{rank}(G_t G_t^\top)$},
				xmin=1, xmax=63.2, 
				ymin=0, ymax=6.5,
				xtick distance=10,
				extra x ticks={1,61,62,63},
				extra x tick labels={$1$,,,},
				extra x tick style={
					major tick length=1.5pt,
					grid=none,
					tick label style={yshift=-0.5ex}
				},
				extra x tick style={major tick length=1.5pt, grid=none},
				minor x tick num=9, 
				ytick distance=1,   
				grid=both, 
				major grid style={dashed, gray!60}, 
				minor grid style={dotted, gray!30}, 
				tick align=outside,
				tick pos=left,
				axis lines*=left
				]
				\addplot[
				color=blue, 
				thick,      
				mark=*,
				mark options={fill=red, draw=red, scale=0.25} 
				] coordinates {
					(1,1) (2,2) (3,3) (4,4) (5,5) (6,6) (7,6) (8,5) (9,6) (10,6) 
					(11,5) (12,4) (13,3) (14,4) (15,5) (16,6) (17,6) (18,6) (19,5) (20,6) 
					(21,6) (22,5) (23,6) (24,6) (25,6) (26,6) (27,6) (28,5) (29,4) (30,5) 
					(31,6) (32,6) (33,5) (34,4) (35,5) (36,6) (37,6) (38,6) (39,6) (40,6) 
					(41,5) (42,6) (43,6) (44,5) (45,6) (46,6) (47,6) (48,5) (49,4) (50,3) 
					(51,4) (52,5) (53,6) (54,6) (55,5) (56,6) (57,6) (58,5) (59,4) (60,3) 
					(61,2) (62,1) (63,0)
				};
			\end{axis}
		\end{tikzpicture}
		\par\vspace{1mm}
		\makebox[\textwidth][c]{\small (c) $n=6$, $f(x)=x^6+x+1$.}
	\end{minipage}\hfill
	\begin{minipage}[t]{0.49\textwidth}
		\centering
		\begin{tikzpicture}
			\begin{axis}[
				width=\textwidth, 
				height=6cm,
				xlabel={$t$},
				ylabel={$\mathrm{rank}(G_t G_t^\top)$},
				xmin=1, xmax=63.2, 
				ymin=0, ymax=6.5,
				xtick distance=10,
				extra x ticks={1,61,62,63},
				extra x tick labels={$1$,,,},
				extra x tick style={
					major tick length=1.5pt,
					grid=none,
					tick label style={yshift=-0.5ex}
				},
				extra x tick style={major tick length=1.5pt, grid=none},
				minor x tick num=9, 
				ytick distance=1,   
				grid=both, 
				major grid style={dashed, gray!60}, 
				minor grid style={dotted, gray!30}, 
				tick align=outside,
				tick pos=left,
				axis lines*=left
				]
				\addplot[
				color=blue, 
				thick,      
				mark=*,
				mark options={fill=red, draw=red, scale=0.25} 
				] coordinates {
					(1,1) (2,2) (3,3) (4,4) (5,5) (6,6) (7,6) (8,6) (9,6) (10,5) 
					(11,6) (12,6) (13,6) (14,5) (15,4) (16,5) (17,6) (18,6) (19,5) (20,6) 
					(21,6) (22,5) (23,6) (24,6) (25,6) (26,5) (27,4) (28,3) (29,4) (30,5) 
					(31,6) (32,6) (33,5) (34,4) (35,3) (36,4) (37,5) (38,6) (39,6) (40,6) 
					(41,5) (42,6) (43,6) (44,5) (45,6) (46,6) (47,5) (48,4) (49,5) (50,6) 
					(51,6) (52,6) (53,5) (54,6) (55,6) (56,6) (57,6) (58,5) (59,4) (60,3) 
					(61,2) (62,1) (63,0)
				};
			\end{axis}
		\end{tikzpicture}
		\par\vspace{1mm}
		\makebox[\textwidth][c]{\small (d) $n=6$, $f(x)=x^6+x^4+x^3+x+1$.}
	\end{minipage}
	
	\caption{Rank evolution for $n=5$ and $n=6$ corresponding to different choices of primitive polynomials.}
	\label{fig:gram_rank_all}
\end{figure}

The main results of this paper can be summarized in four structural properties of binary m-sequences. The following theorem gives a global rank distribution of Gram matrices arising from m-sequences.

\begin{theorem}[Rank distribution]\label{thm:rank_distribution}
Let $n\ge 3$, and let $\{s_t\}_{t=0}^{\infty}$ be an m-sequence of order $n$. For $1\le t\le 2^{n}-1$, let $G_{t}$ be the observability matrix associated with the m-sequence. Then the rank distribution of the family $\{G_t G_t^\top \mid 1 \le t \le 2^n-1\}$ is given by
\begin{equation}
	\#\{\,1\le t\le 2^n-1\mid \ \rank(G_tG_t^\top)=k\,\}=
	\begin{cases}
		1, & k=0,\\
		2, & k=1,\\
		2^{k-1}, & 2\le k\le n-1,\\
		2^{n-1}-2, & k=n.
	\end{cases}
\end{equation}
\end{theorem}

The next three theorems characterize the local rank dynamics.

\begin{theorem}[Persistence after reaching full rank]\label{thm Persistence}
Let $n\ge 3$, write $r_n(t):=\rank(G_tG_t^\top)$, and regard $t$ modulo $2^n-1$. If $r_n(t-1)=n-1$ and $r_n(t)=n$, then $$r_n(t+1)=n.$$
	
\end{theorem}

\begin{theorem}[Instability of rank-deficient states]\label{thm Insta}
Let $n\ge 3$, if $r_n(t)<n$, then $$r_n(t+1)\ne r_n(t).$$
	
\end{theorem}

\begin{theorem}[Enumeration of local minima]\label{thm Enumeration of local minima}
	A value $t$ (modulo $2^n-1$) is called a local minimum of the rank function $r_n(\cdot)$ if 
	\[
	r_n(t-1)=r_n(t)+1
	\quad\text{and}\quad
	r_n(t+1)=r_n(t)+1.
	\]
For $n\ge 3$, the number of local minima of $r_n(\cdot)$ is
	\[
	\frac{2^{n-1}-(-1)^{n-1}}{3}.
	\]
\end{theorem}


This paper is organized as follows. Section \ref{sec-preliminary} introduces the main tools and some auxiliary results. Section \ref{sec-main result} forms the technical core of the paper: we characterize the singular Gram matrices by rational functions, derive an explicit rank formula for all singular matrices, and then prove the complete rank distribution via generating functions. Building on these results, Section \ref{sec-dynamics} investigates the local dynamics of the rank function along the natural cyclic ordering, including the instability of rank-deficient states, persistence after reaching full rank, and the enumeration of local minima. Section \ref{sec-LCD} turns to coding-theoretic applications. Interpreting the code family generated by the observability matrices as a punctured family of binary cyclic simplex codes, we characterize the LCD members and determine the corresponding hull distribution. Finally, Section \ref{sec-conclu} concludes the paper with a brief summary and some open problems. The logical relations among key ingredients and main results of the paper are summarized in Figure \ref{fig:intro-roadmap}.



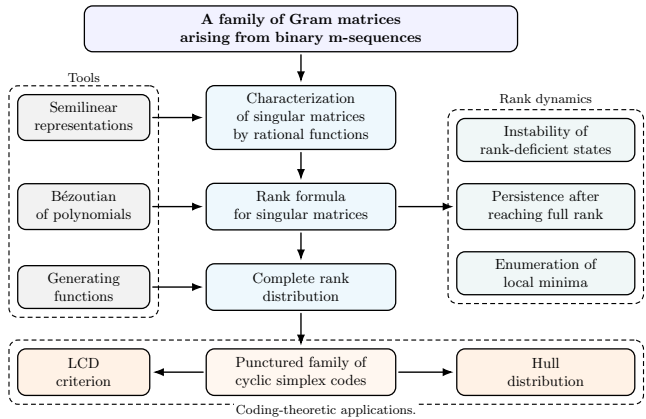
\begin{figure}[htbp]
	\centering
	\resizebox{0.65\textwidth}{!}{%
		\begin{tikzpicture}[
			>=Latex,
			font=\small,
			node distance=5.5mm and 8mm,
			box/.style={
				draw,
				rounded corners=1.8mm,
				line width=0.8pt,
				align=center,
				inner xsep=6pt,
				inner ysep=4.5pt,
				fill=white
			},
			mainbox/.style={
				box,
				fill=blue!5,
				text width=9.0cm
			},
			corebox/.style={
				box,
				fill=cyan!6,
				text width=3.7cm,
				minimum height=10mm
			},
			tool/.style={
				box,
				fill=gray!10,
				text width=2.45cm,
				minimum height=9mm
			},
			dyn/.style={
				box,
				fill=teal!6,
				text width=2.8cm,
				minimum height=9mm
			},
			codefam/.style={
				box,
				fill=orange!6,
				minimum height=10mm
			},
			code/.style={
				box,
				fill=orange!10,
				minimum height=10mm
			},
			group/.style={
				draw,
				dashed,
				dash pattern=on 3pt off 2pt,
				rounded corners=2mm,
				inner sep=5pt,
				line width=0.7pt
			},
			arr/.style={
				->,
				line width=0.85pt,
				draw=black,
				line cap=round,
				line join=round,
				shorten >=1.5pt,
				shorten <=1.5pt
			}
			]
			
			\node[mainbox] (main) {%
				\textbf{A family of Gram matrices}\\
				\textbf{arising from binary m-sequences}
				
			};
			
			\node[corebox, below=7mm of main] (sing) {%
				Characterization\\
				of singular matrices\\
				by rational functions\\
				
			};
			
			\node[corebox, below=7mm of sing] (rankf) {%
				Rank formula\\
				for singular matrices
				
			};
			
			\node[corebox, below=7mm of rankf] (dist) {%
				Complete rank\\
				distribution
			};
			
			\node[tool, left=12mm of sing] (gd) {Semilinear \\representations};
			\node[tool, left=12mm of rankf] (bez) {B{\'e}zoutian\\ of polynomials};
			\node[tool, left=12mm of dist] (gen) {Generating\\functions};
			
			\node[group, fit=(gd)(bez)(gen),
			label={[font=\footnotesize, fill=white, inner sep=1pt, yshift=1pt]above:Tools}] (toolgrp) {};
			
			\node[dyn, text width=3.45cm, right=13mm of rankf] (pers) {%
				Persistence after\\
				reaching full rank
			};
			
			\node[dyn, text width=3.45cm, above=4.5mm of pers] (inst) {%
				Instability of\\
				rank-deficient states
			};
			
			\node[dyn, text width=3.45cm, below=4.5mm of pers] (loc) {%
				Enumeration of\\
				local minima
			};
			
			\node[group, fit=(inst)(pers)(loc),
			label={[font=\footnotesize, fill=white, inner sep=1pt, yshift=1pt]above:Rank dynamics}] (dyngrp) {};
			
			\node[codefam, text width=3.7cm, below=8mm of dist] (pcodes) {%
				Punctured family of\\
			cyclic simplex codes
			};
			
			\node[code, text width=2.45cm] (lcd) at (gd.center |- pcodes.center) {%
				LCD\\
				criterion
			};
			
			\node[code, text width=3.45cm] (hull) at (pers.center |- pcodes.center) {%
				Hull\\
				distribution
			};
			
			\node[group, fit=(lcd)(pcodes)(hull),
			label={[font=\footnotesize, fill=white, inner sep=1pt, yshift=1pt]below:Coding-theoretic applications.}] (codegrp) {};
			
			\draw[arr] (main.south) -- (sing.north);
			
			\draw[arr] (gd.east) -- (sing.west);
			\draw[arr] (bez.east) -- (rankf.west);
			\draw[arr] (gen.east) -- (dist.west);
			
			\draw[arr] (sing.south) -- (rankf.north);
			\draw[arr] (rankf.south) -- (dist.north);
			
			\draw[arr] (rankf.east) -- (pers.west);
			
			\draw[arr] (dist.south) -- (pcodes.north);
			\draw[arr] (pcodes.west) -- (lcd.east);
			\draw[arr] (pcodes.east) -- (hull.west);
			
	\end{tikzpicture}}
	\caption{Overview of the main results and tools}
	\label{fig:intro-roadmap}
\end{figure}

\section{Preliminaries}\label{sec-preliminary}
In this section, we present some preliminaries and auxiliary results that will be used.

\subsection{LFSR sequences and m-sequences}

In this subsection, we provide some basic properties of binary LFSR sequences and m-sequences.

\begin{definition}
	Let $g(x)\in \mathbb F_2[x]$ with $g(0)=1$. The order of $g(x)$ is defined to be the least positive integer $N$ such that $g(x)\mid (x^N-1)$.
\end{definition}

\begin{definition}
	For the field extension $\mathbb F_{2^n}/\mathbb F_2$, the trace function
	\[
	\Tr_{\mathbb F_{2^n}/\mathbb F_2}:\mathbb F_{2^n}\to \mathbb F_2
	\]
	is defined by
	\[
	\Tr_{\mathbb F_{2^n}/\mathbb F_2}(a)=a+a^2+\cdots+a^{2^{n-1}}.
	\]

\end{definition}

\begin{definition}
	Let
	\[
	g(x)=x^n+c_{n-1}x^{n-1}+\cdots+c_1x+c_0\in \mathbb F_2[x].
	\]
	A binary sequence $\{s_t\}_{t=0}^\infty$ is called an LFSR sequence with characteristic polynomial $g(x)$ if
	\[
	s_{t+n}=c_{n-1}s_{t+n-1}+\cdots+c_1s_{t+1}+c_0s_t
	\quad\text{for all } t\ge 0.
	\]
	 If $\{s_t\}_{t=0}^\infty$ cannot be generated by any LFSR of smaller order, then $g(x)$ is called the minimal polynomial of the sequence. The degree of the minimal polynomial of $\{s_t\}_{t=0}^{\infty}$ is called its linear complexity.
\end{definition}
The minimal polynomial and linear complexity of an LFSR sequence can be determined by the following proposition \cite{Lidl 1997}.

\begin{proposition}
Let $\{s_t\}_{t=0}^\infty$ be an LFSR sequence of period $N$, and let $S^{N}(x)=s_0+s_1x+\cdots+s_{N-1}x^{N-1}$. Then 
	\begin{enumerate}
	\item  the minimal polynomial of $\{s_t\}_{t=0}^\infty$ is 
	$$\frac{x^N-1}{\gcd(x^N-1,S^{N}(x))};$$
	
	\item the linear complexity of $\{s_t\}_{t=0}^\infty$ is  $$n-\deg(\gcd(x^N-1,S^{N}(x))).$$
\end{enumerate}
\end{proposition}

Linear complexity serves as a crucial measure of the strength of keystream sequences in stream cipher systems, and more broadly, as an indicator of the randomness quality of sequences \cite{Ding 1997 linear complexity}.

The trace representation and minimal period of LFSR sequences are presented in the following theorem, summarized from \cite{Golomb 1967, Goresky 2012}.

\begin{theorem}\label{thm:lfsr-trace}
	Let $\{s_t\}_{t=0}^\infty$ be a binary LFSR sequence with characteristic polynomial $g(x)$.
	
	\begin{enumerate}
		\item Suppose that $g(x)\in \mathbb F_2[x]$ is irreducible of degree $n$, and let $\alpha\in \mathbb F_{2^n}$ be a root of $g(x)$. Then there exists $\beta\in \mathbb F_{2^n}$, such that
		\[
		s_t=\Tr_{\mathbb F_{2^n}/\mathbb F_2}(\beta\alpha^t).
		\]
		If $g(x)$ is the minimal polynomial of the sequence, then the minimal period of $\{s_t\}_{t=0}^\infty$ is $\ord(g)$.
		
		\item More generally, suppose that
		\[
		g(x)=g_1(x)g_2(x)\cdots g_e(x),
		\]
		where $g_1(x),\dots,g_e(x)\in \mathbb F_2[x]$ are distinct irreducible polynomials, with $\deg(g_i)=d_i$, and let $\alpha_i\in \mathbb F_{2^{d_i}}$ be a root of $g_i(x)$ for each $1\le i\le e$. Then there exist $\gamma_i\in \mathbb F_{2^{d_i}}$ such that
		\[
		s_t=\Tr\!\left(\sum_{i=1}^e \gamma_i\alpha_i^t\right).
		\]
		where $\Tr$ denotes the trace function from the splitting field of $g(x)$ to $\mathbb F_2$. If $g(x)$ is the minimal polynomial of the sequence, then the minimal period of $\{s_t\}_{t=0}^\infty$ is
		\[
		\mathrm{lcm}\bigl(\ord(g_1),\dots,\ord(g_e)\bigr).
		\]
	\end{enumerate}
\end{theorem}

 If $g(x)$ is primitive of degree $n$, then $\ord(g)=2^n-1$, which is the maximum possible period of a nonzero binary LFSR sequence of order $n$. Let $\alpha\in \mathbb F_{2^n}$ be a root of $g(x)$. Then such an m-sequence can be written in the form
\begin{equation}\label{eq:msequence-trace}
	s_t=\Tr_{\mathbb F_{2^n}/\mathbb F_2}(\lambda\alpha^t),\quad \lambda\in \mathbb F_{2^n}^*.
\end{equation}

We further recall the generating function representation of LFSR sequences,
a classical result that expresses an LFSR sequence as a rational function. For details, the reader is referred to \cite{Goresky 2012}.

\begin{proposition}\label{prop:gen-fun}
Let	$\{s_t\}_{t=0}^{\infty}$ be a binary LFSR sequence with characteristic polynomial
	\[
	g(z) \;=\; z^{r} + c_{r-1} z^{r-1} + \cdots + c_{1} z + c_{0}
	\;\in\; \mathbb{F}_2[z], \quad c_{0} = 1
	\]
	 and initial state
	$(s_0, s_1, \ldots, s_{r-1}) \in \mathbb{F}_2^{r}$. Then we have the representation:
	\[
	\frac{h(z)}{\widetilde{g}(z)}=\sum_{t=0}^{\infty} s_t z^{t},
	\]
where $\widetilde{g}(z) = z^{r} g(z^{-1})$ is the reciprocal polynomial of
	$g(z)$, and
	\[
	h(z) \;=\; \sum_{i=0}^{r-1}
	\bigg(\sum_{j=0}^{i} c_{\,r-i+j}\, s_{j}\bigg) z^{i}.\]
\end{proposition}

\begin{example}
	Let $n = 5$ and take a primitive polynomial
	\[
	g(z) \;=\; z^{5} + z^{2} + 1 \;\in\; \mathbb{F}_{2}[z],
	\]
	so that its reciprocal is $\widetilde{g}(z) = 1 + z^{3} + z^{5}$. The
	coefficients of $g(z)$ are
	$(c_{0},c_{1},c_{2},c_{3},c_{4},c_{5}) = (1,0,1,0,0,1)$. With initial state
	$(s_{0},s_{1},s_{2},s_{3},s_{4}) = (1,0,0,0,0)$, the linear recurrence
	$s_{t+5} = s_{t+2} + s_{t}$ generates an m-sequence of order $5$; and we have
	
 $$h(z) = \sum_{i=0}^{4} \bigl(\sum_{j=0}^{i} c_{5-i+j}\, s_{j}\bigr) z^{i}= 1 + z^{3}.$$
  
  Then Proposition~\ref{prop:gen-fun} gives
	\[
	\sum_{t=0}^{\infty} s_{t}\, z^{t}
	\;=\; \frac{1 + z^{3}}{1 + z^{3} + z^{5}}.
	\]
\end{example}


\subsection{Semilinear representation of Galois groups}
\label{subsec:semilinear}

In this subsection we introduce the semilinear representation of Galois groups, which is needed in Section \ref{sec-main result}. Semilinear representation is a natural extension of linear representation, and frequently appears in several branches of modern mathematics, including projective geometry~\cite{Artin 1957}, Galois cohomology~\cite{Serre 2002}, and $p$-adic Hodge theory \cite{Fontaine 1994}.


We first give the definitions of semilinear map and semilinear representation.

\begin{definition}
Let $V$ and $W$ be vector spaces over a field $K$, a semilinear map between them is a function $T: V\to W$ satisfying

\begin{enumerate}
	\item $T(u+v)=T(u)+T(v)$;
	
	\item there exists a field automorphism $\psi $ of $K$ such that, $T(\lambda v)=\psi(\lambda)T(v), \forall \lambda\in K.$

\end{enumerate}
Then $T$ is called $\psi$-semilinear. When $V=W$, we say that $T$ is a semilinear transformation on $V$. All invertible semilinear transforms on a vector space $V$ forms a group $\Gamma\rm{L}(V)$, called the general semilinear group, extending the general linear group $\rm{GL}(V)$.
\end{definition}

\begin{definition}
A  semilinear representation of a group $G$ on a vector space $V$ over a field $K$ is a group homomorphism $\rho : G\to \Gamma\rm{L}(V)$ from $G$ to $\Gamma\rm{L}(V)$ given by $g\mapsto \rho_{g} $, such that $\rho_{g_{1}g_{2}}=\rho_{g_{1}}\rho_{g_{2}}$.
\end{definition}

We will use the following well-known lemma, which asserts that distinct characters are linear independent.

\begin{lemma}[Dedekind's lemma]\label{lem:dedekind}
	Let $G$ be a group, $L$ be a field, and let $\chi_1, \dots, \chi_n : G \to L^{\times}$ be pairwise distinct $L$-characters of $G$. If $\lambda_1, \dots, \lambda_n \in L$ satisfy
	$$
	\sum_{i=1}^{n}\lambda_i \chi_i(g) = 0, \quad \forall  g \in G,
	$$
	then $\lambda_i =0, \, 1\le i\le n$.
\end{lemma}


For a Galois extension $L/K$, if the Galois group $G$ admits a semilinear representation $\rho : G\to \Gamma\rm{L}(V)$, $g\mapsto \rho_{g} $, where $V$ is an $L$-vector space, then the set of fixed points: i.e., $$V^G = \{v \in V \mid \rho_{\sigma}(v) = v, \quad \forall \sigma \in G\},$$ forms a $K$-subspace of $V$. Using the technique of scalar extension, we can further determine the dimension of $V^G$ by the following theorem.

Given a field extension $L/K$ and
a $K$-vector space $W$, it is well-known that the extension of scalars $W\otimes_{K}L$ is naturally an $L$-vector space with $\dim_{L}(W\otimes_{K}L)
=\dim_{K}W$. 


\begin{theorem}\label{thm:galois-descent}
	Let $L/K$ be a finite Galois extension with Galois group $G = \mathrm{Gal}(L/K)$, and let
	$V$ be an $L$-vector space equipped with a semilinear representation $\rho : G\to \Gamma\rm{L}(V)$, $g\mapsto \rho_{g} $. Define a $K$-subspace of $V$:  $$V^G = \{v \in V \mid \rho_{\sigma}(v) = v, \, \forall \sigma \in G\}.$$ Then the
	map
	\[
	\mu \colon V^{G}\otimes_{K}L \longrightarrow V,
	\quad
	v\otimes \lambda \longmapsto \lambda v,
	\]
	is an isomorphism of $L$-vector spaces. In particular,
	$\dim_{K}V^{G}=\dim_{L}V$.
\end{theorem}

\begin{proof}
We first prove the injectivity. Suppose, for contradiction, that $\ker(\mu)\neq 0$, and pick a non-zero element
	\[
	\omega=\sum_{i=1}^{r}v_{i}\otimes \lambda_{i}\in\ker(\mu)
	\]
	with $r$ chosen as the smallest. Clearly $\lambda_{i}\neq 0$ for each $i$. Dividing $\omega$ by $\lambda_{1}$,
	we may further assume $\lambda_{1}=1$. The condition $\mu(\omega)=0$
	reads
	\begin{equation}\label{eq:inj-star}
		\sum_{i=1}^{r}\lambda_{i}v_{i}=0.
	\end{equation}
	Applying any $\sigma\in G$ to Eq. \eqref{eq:inj-star} and together with $\sigma(v_{i})=v_{i}$, we obtain
	\begin{equation}\label{eq:inj-starstar}
		\sum_{i=1}^{r}\sigma(\lambda_{i})\,v_{i}=0, \quad \forall\,\sigma\in G.
	\end{equation}
	Subtracting Eq. \eqref{eq:inj-star} from Eq. \eqref{eq:inj-starstar} yields
	\begin{equation}\label{eq:inj-dagger}
		\sum_{i=2}^{r}\bigl(\sigma(\lambda_{i})-\lambda_{i}\bigr)\,v_{i}=0,
		\quad \forall\,\sigma\in G.
	\end{equation}
	
	Suppose that some $\lambda_{i_{0}}$ with
	$i_{0}\geq 2$ does not lie in $K=L^{G}$. Then there exists
	$\sigma_{0}\in G$ with
	$\sigma_{0}(\lambda_{i_{0}})\neq \lambda_{i_{0}}$. Specializing Eq.
	\eqref{eq:inj-dagger} to $\sigma=\sigma_{0}$ gives that the element
	\[
	\omega':=\sum_{i=2}^{r}v_{i}\otimes
	\bigl(\sigma_{0}(\lambda_{i})-\lambda_{i}\bigr) \in V^{G}\otimes_{K}L
	\]
	satisfies $\mu(\omega')=0$ and has its $i_{0}$-th coefficient non-zero,
	so $\omega'\neq 0$; but this contradicts the
	minimality of $r$.
	
	Therefore $\lambda_{i}\in K$ for every $i$, together with Eq. \eqref{eq:inj-star}, we obtain
	\[
	\omega=\sum_{i=1}^{r}v_{i}\otimes \lambda_{i}
	=\sum_{i=1}^{r}(\lambda_{i}v_{i})\otimes 1
	=\Bigl(\sum_{i=1}^{r}\lambda_{i}v_{i}\Bigr)\otimes 1=0,
	\]

contradicting the assumption $\omega\neq 0$. Hence $\ker(\mu)=0$.

Next, we prove the surjectivity of $\mu$. Let $W\subseteq V$ be the image of $\mu$, which is the $L$-span of $V^{G}$. We prove that $W=V$.
	
Observe that $W$ is stable under the $G$-action on $V$, i.e., $\rho_{\sigma}(W)\subseteq W$. Indeed, for any
$\sigma\in G$, and any element $\sum_{j}\ell_{j}w_{j}$ of $W$ with $\ell_{j}\in L$ and $w_{j}\in V^{G}$, we have
\[
\rho_{\sigma}\Bigl(\sum_{j}\ell_{j}w_{j}\Bigr)
=\sum_{j}\rho_{\sigma}(\ell_{j}w_{j})
=\sum_{j}\sigma(\ell_{j})\,w_{j}\in W.
\]

Since $W$ is an $L$-subspace of $V$, the quotient $V/W$ is an
$L$-vector space. Since $W$ is a $G$-stable $L$-subspace of $V$, it is a subrepresentation
of $V$. The quotient $V/W$ therefore inherits the structure of a
semilinear $G$-representation, the quotient representation,
whose action $\bar\rho\colon G\to \Gamma L(V/W)$ is uniquely determined
by the commutativity of
\[
\begin{tikzcd}
	V \arrow[r, "\rho_{\sigma}"] \arrow[d, "\pi"'] & V \arrow[d, "\pi"] \\
	V/W \arrow[r, "\bar\rho_{\sigma}"'] & V/W
\end{tikzcd}
\]

for every $\sigma\in G$, where $\pi\colon V\to V/W$ is the canonical
projection. Explicitly, $\bar\rho_{\sigma}(\overline{v})
=\overline{\rho_{\sigma}(v)}$, and the semilinearity relation
$\bar\rho_{\sigma}(\ell\,\overline{v})=\sigma(\ell)\,
\bar\rho_{\sigma}(\overline{v})$ descends from the corresponding
relation on $V$.

Now suppose, for contradiction, that $V/W\neq 0$, and fix any
$u\in V\setminus W$, so that $\overline{u}\neq 0$ in $V/W$. For each
$\lambda\in L$, consider
\[
\widetilde{T}(\lambda):=\sum_{\sigma\in G}\rho_{\sigma}(\lambda u)
\in V.
\]
For every $\eta\in G$, we have
\[
\rho_{\eta}\bigl(\widetilde{T}(\lambda)\bigr)
=\sum_{\sigma\in G}\rho_{\eta\sigma}(\lambda u)
=\widetilde{T}(\lambda),
\]
so $\widetilde{T}(\lambda)\in V^{G}\subseteq W$. Passing to $V/W$ and
using semilinearity of $\bar\rho_{\sigma}$, the image
$\overline{\widetilde{T}(\lambda)}=0$ reads
\begin{equation}\label{eq:surj-star}
	\sum_{\sigma\in G}\sigma(\lambda)\,\bar\rho_{\sigma}(\overline{u})=0
	\quad\text{in } V/W, \quad \forall\,\lambda\in L.
\end{equation}

Fix an $L$-basis $\{\bar e_{k}\}_{k\in\mathcal K}$ of $V/W$, and
expand
\[
\bar\rho_{\sigma}(\overline{u})=\sum_{k\in\mathcal K}c_{\sigma,k}\,
\bar e_{k}, \quad c_{\sigma,k}\in L,
\]
where for each fixed $\sigma$ only finitely many $c_{\sigma,k}$ are
non-zero. Substituting into Eq. \eqref{eq:surj-star} and reading off the
coefficient of $\bar e_{k}$ gives, for every $k\in\mathcal K$ and
every $\lambda\in L$,
\begin{equation}\label{eq:surj-char}
	\sum_{\sigma\in G}c_{\sigma,k}\,\sigma(\lambda)=0.
\end{equation}

For each $\sigma\in G$ we regard the
restriction
\[
\chi_{\sigma}:=\sigma|_{L^{\times}}\colon L^{\times}\longrightarrow L^{\times}
\]
as an $L$-character of the group $L^{\times}$, then distinct $\sigma\in G$ yield distinct characters
$\chi_{\sigma}$. Restricting Eq. \eqref{eq:surj-char} to
$\lambda\in L^{\times}$ yields
\[
\sum_{\sigma\in G}c_{\sigma,k}\,\chi_{\sigma}(\lambda)=0,
\quad \forall\,\lambda\in L^{\times}.
\]
 Applying
Lemma~\ref{lem:dedekind} forces $c_{\sigma,k}=0$ for every
$\sigma\in G$ and every $k\in\mathcal K$. Hence
$\bar\rho_{\sigma}(\overline{u})=0$ in $V/W$ for every $\sigma\in G$. Taking $\sigma=\mathrm{id}_{L}$, so that
$\bar\rho_{\mathrm{id}_{L}}=\mathrm{id}_{V/W}$, yields
$\overline{u}=0$, contradicting the choice of $u\in V\setminus W$.
Hence $V=W$, and $\mu$ is surjective.

\end{proof}

\subsection{B{\'e}zoutian of polynomials}

We briefly recall the definition of the B{\'e}zoutian and the rank formula that will be needed. B{\'e}zoutians is a classical object invented by B{\'e}zout \cite{Bezout 1764,Bezout 1779}, and developed by Sylvester \cite{Sylvester 1904}, Jacobi \cite{Jacobi 1836} and Cayley \cite{Cayley 1857}, with primary applications in elimination theory. Over the years, B{\'e}zoutian has become a useful tool in many fields, such as control theory, and symbolical computing \cite{Lascoux 2005}.

Let $K$ be a field, and let $f(x),g(x)\in K[x]$ satisfy
\[
\deg(f),\deg(g)\le n.
\]
The B{\'e}zoutian of order $n$ associated with $f$ and $g$ is a symmetric matrix $B_{n}(f,g)=(b_{i,j})_{0\le i,j\le n-1}$, where the entries $b_{i,j}$ are coefficients in the following expansion of a bivariate polynomial:
\[
\frac{f(x)g(y)-g(x)f(y)}{x-y}=\sum_{i=0}^{n-1}\sum_{j=0}^{n-1} b_{i,j} x^i y^j.
\]

The following rank formula of B{\'e}zoutian is standard.

\begin{theorem}\label{lem:bezout-rank}
	Let $f(z),g(z)\in K[z]$ be nonzero polynomials with
	\[
	\max\{\deg(f),\deg(g)\}\le n.
	\]
	Let $B_{n}(f,g)$ be their associated B{\'e}zoutian.
	Then
	\[
	\rank\bigl(B_{n}(f,g))\bigr)
	=
	\max\{\deg(f),\deg(g)\}-\deg(\gcd(f,g)).
	\]
\end{theorem}

\begin{example}
	Let $K=\mathbb F_2$, and let
	\[
	f(z)=z^3+1,\quad g(z)=z^3+z^2.
	\]
	Then
	\[
	\gcd(f,g)=z+1,
	\]

	\[
	\begin{aligned}
		\frac{f(x)g(y)-g(x)f(y)}{x-y}
		&=\frac{(x^3+1)(y^3+y^2)-(x^3+x^2)(y^3+1)}{x-y}\\
		&=\frac{x^3y^2+y^3+y^2+x^2y^3+x^2}{x-y}\\
		&=x^2y^2+x^2+xy+x+y^2+y.
	\end{aligned}
	\]
	Hence the associated B\'ezoutian matrix is
	\[
	B_3(f,g)=
	\begin{pmatrix}
		0 & 1 & 1\\
		1 & 1 & 0\\
		1 & 0 & 1
	\end{pmatrix}.
	\]
Therefore,
	\[
	\rank\bigl(B_3(f,g)\bigr)=2,
	\]
	in agreement with Theorem \ref{lem:bezout-rank}.
\end{example}

\section{Rank problem of Gram matrices associated with m-sequences}\label{sec-main result}

In this section, we explicitly determine the ranks of the singular Gram matrices $G_t G_t^\top$, associated with binary m-sequences over the full period $1 \le t \le 2^n - 1$. Via a factorization for the observability matrix $G_t$, we first lift the rank problem from the binary field to its extension field $\mathbb{F}_{2^n}$. Subsequently, for singular Gram matrices, we develop a rank characterization using semilinear representation of Galois groups. Employing a rank argument of polynomial B{\'e}zoutian, we then obtain a rank formula for those matrices. This allows us to translate the rank distribution problem into enumeration of certain rational functions.

\subsection{Reduction to symmetric matrices over $\mathbb{F}_{2^n}[x]$}

Let $\{s_t\}_{t\ge 0}$ be the binary $m$-sequence given by
\[
s_t=\Tr_{\mathbb F_{2^n}/\mathbb F_2}(\lambda\alpha^t),
\]
where $\alpha$ is a primitive element of $\mathbb F_{2^n}$ and $\lambda\in \mathbb F_{2^n}^*$. For
$1\le t\le 2^n-1$, recall that the associated observability matrix is
\begin{equation}
G_t=
\begin{pmatrix}
	s_0 & s_1 & \cdots & s_{t-1}\\
	s_1 & s_2 & \cdots & s_t\\
	\vdots & \vdots & \ddots & \vdots\\
	s_{n-1} & s_n & \cdots & s_{n+t-2}
\end{pmatrix}\in \mathbb F_2^{\,n\times t}.
\end{equation}
For vectors $\mathbf{u}=(u_1,\dots,u_t)$ and $\mathbf{v}=(v_1,\dots,v_t)$ in $\mathbb F_2^t$, define the Euclidean inner product of $\mathbf{u}$ and $\mathbf{v}$ by
\[
\mathbf{u}\cdot \mathbf{v}=\sum_{r=1}^t u_rv_r \in \mathbb F_2.
\]
Write
\[
\mathbf{R}_i=(s_{i-1},s_i,\dots,s_{i+t-2})\in \mathbb F_2^t,\quad 1\le i\le n,
\]
for the $i$-th row of $G_t$. Our goal is to determine the rank of the following Gram matrix:
\[
G_tG_t^\top=
\begin{pmatrix}
	\mathbf{R}_1\cdot \mathbf{R}_1 & \mathbf{R}_1\cdot \mathbf{R}_2 & \cdots & \mathbf{R}_1\cdot \mathbf{R}_n\\
	\mathbf{R}_2\cdot \mathbf{R}_1 & \mathbf{R}_2\cdot \mathbf{R}_2 & \cdots & \mathbf{R}_2\cdot \mathbf{R}_n\\
	\vdots & \vdots & \ddots & \vdots\\
	\mathbf{R}_n\cdot \mathbf{R}_1 & \mathbf{R}_n\cdot \mathbf{R}_2 & \cdots & \mathbf{R}_n\cdot \mathbf{R}_n
\end{pmatrix}.
\]

One obstacle is that, over the binary field, a direct rank analysis of $G_tG_t^\top$ is not tractable. To obtain a more
explicit description, we lift the problem to the extension field $\mathbb F_{2^n}$ by the following lemma, where $G_t$ admits a
Vandermonde-type factorization. 
\begin{lemma}\label{lem:decomposition-Gt}
For $1\le t\le 2^{n}-1$, let $G_{t}$ be the observability matrix associated with a binary $m$-sequence given by
\[
s_t=\Tr_{\mathbb F_{2^n}/\mathbb F_2}(\lambda\alpha^t).
\] For $1\le \ell\le n$, write $\alpha_\ell=\alpha^{2^{\ell-1}}$ and $\lambda_\ell=\lambda^{2^{\ell-1}}$. Then, 
\[
G_t=V\Lambda \widetilde G_t,
\]
where $V_{i,\ell}=\alpha_\ell^{\,i-1}$ for $1\le i,\ell\le n$, $\Lambda=\operatorname{diag}(\lambda_1,\dots,\lambda_n)$, and
\begin{equation}
\widetilde G_t=
\begin{pmatrix}
	1 & \alpha_1 & \alpha_1^2 & \cdots & \alpha_1^{t-1}\\
	1 & \alpha_2 & \alpha_2^2 & \cdots & \alpha_2^{t-1}\\
	\vdots & \vdots & \vdots & \ddots & \vdots\\
	1 & \alpha_n & \alpha_n^2 & \cdots & \alpha_n^{t-1}
\end{pmatrix}.
\end{equation}
	
Furthermore, one has
	\begin{enumerate}
		\item $\rank(G_t)=\rank(\widetilde G_t)=\min\{t,n\}$.
		
		\item Let $x\in\mathbb F_{2^n}^*$ be a variable, and let $M(x)$ be a symmetric matrix over $\mathbb{F}_{2^n}[x]$ with 
		\[
		M(x)_{i,j}=\frac{1+x^{2^{i-1}+2^{j-1}}}{1+\alpha^{2^{i-1}+2^{j-1}}}, \quad 1\le i,j\le n, \quad i.e., \]
 \begin{equation}\label{eq M}
M(x)=
\begin{pmatrix}
	\dfrac{1+x^{2^0+2^0}}{1+\alpha^{2^0+2^0}}
	&
	\dfrac{1+x^{2^0+2^1}}{1+\alpha^{2^0+2^1}}
	&
	\cdots
	&
	\dfrac{1+x^{2^0+2^{n-1}}}{1+\alpha^{2^0+2^{n-1}}}
	\\[1.2ex]
	\dfrac{1+x^{2^1+2^0}}{1+\alpha^{2^1+2^0}}
	&
	\dfrac{1+x^{2^1+2^1}}{1+\alpha^{2^1+2^1}}
	&
	\cdots
	&
	\dfrac{1+x^{2^1+2^{n-1}}}{1+\alpha^{2^1+2^{n-1}}}
	\\
	\vdots & \vdots & \ddots & \vdots
	\\
	\dfrac{1+x^{2^{n-1}+2^0}}{1+\alpha^{2^{n-1}+2^0}}
	&
	\dfrac{1+x^{2^{n-1}+2^1}}{1+\alpha^{2^{n-1}+2^1}}
	&
	\cdots
	&
	\dfrac{1+x^{2^{n-1}+2^{n-1}}}{1+\alpha^{2^{n-1}+2^{n-1}}}
\end{pmatrix}.
\end{equation}
		If $x=\alpha^t$, where $1\le t\le 2^{n}-1$, then, 
		\[
		\rank(G_tG_t^\top)=\rank(M(x)).
		\]
	\end{enumerate}
\end{lemma}

\begin{proof}
	For $1\le i\le n$ and $1\le j\le t$,
	\[
	(G_t)_{i,j}
	=s_{i+j-2}
	=\Tr_{\mathbb F_{2^n}/\mathbb F_2}(\lambda\alpha^{i+j-2})
	=\sum_{\ell=1}^n(\lambda\alpha^{i+j-2})^{2^{\ell-1}}
	=\sum_{\ell=1}^n \lambda_\ell \alpha_\ell^{\,i-1}\alpha_\ell^{\,j-1},
	\]
	so $G_t=V\Lambda\widetilde G_t$. Since $\alpha_1,\dots,\alpha_n$ are pairwise distinct, $V$ is
	nonsingular; so is $\Lambda$ because $\lambda\ne 0$. Hence
	$\rank(G_t)=\rank(\widetilde G_t)=\min\{t,n\}$.
	
Moreover,
\[
G_tG_t^\top
=
V\Lambda\,\widetilde G_t\widetilde G_t^\top\,\Lambda V^\top,
\]
and hence
\[
\rank(G_tG_t^\top)=\rank(\widetilde G_t\widetilde G_t^\top).
\]
Also,
\[
(\widetilde G_t\widetilde G_t^\top)_{i,j}
=
\sum_{k=1}^t(\widetilde G_t)_{i,k}(\widetilde G_t)_{j,k}
=
\sum_{k=0}^{t-1}(\alpha_i\alpha_j)^k
=
\frac{1+(\alpha_i\alpha_j)^t}{1+\alpha_i\alpha_j},
\quad 1\le i,j\le n.
\]
Since $\alpha_i=\alpha^{2^{i-1}}$, writing $x=\alpha^t$ yields
\[
(\widetilde G_t\widetilde G_t^\top)_{i,j}
=
\frac{1+x^{2^{i-1}+2^{j-1}}}{1+\alpha^{2^{i-1}+2^{j-1}}}
=
M(x)_{i,j},
\quad 1\le i,j\le n.
\]
Thus $\widetilde G_t\widetilde G_t^\top=M(x)$, and hence
\[
\rank(G_tG_t^\top)=\rank(M(x)).
\]
\end{proof}

 Therefore, determining the rank distribution of
 \[
 \{G_tG_t^\top\mid 1\le t\le 2^n-1\}
 \]
 is reduced to that of
 \[
 \{M(x)\mid x\in\mathbb F_{2^n}^*\}.
 \]
 Accordingly, the remainder of this section is devoted to the rank of $M(x)$.

\subsection{Zeros of the determinant polynomials of the associated symmetric matrices}

In this subsection, for $M(x)\in\mathbb F_{2^n}^{n\times n}$ defined as E.q \eqref{eq M}, we study zeros of its determinant polynomial, i.e., the set $$\Omega (M):=\{ x\in\mathbb F_{2^n}^*\mid  \det (M(x))=0\}.$$

By Lemma~\ref{lem:decomposition-Gt}, under the parametrization $x=\alpha^t$, the elements of
$\Omega(M)$ are in one-to-one correspondence with those integers $t\in\{1,2,\dots,2^n-1\}$ for
which the Gram matrix $G_tG_t^\top$ is singular. We characterize the elements in $\Omega (M)$ via a special class of rational functions associated with reciprocal polynomials. The following definition will be used.

\begin{definition}
	Let $\mathbb{F}_2[z]_{<n}$ denote the vector space of polynomials over $\mathbb{F}_2$ with degree strictly less than $n$. 
	For any polynomial $f(z) \in \mathbb{F}_2[z]$, its \textit{standard reciprocal polynomial} is defined as $\tilde{f}(z) = z^{\deg(f)} f(z^{-1})$. A polynomial is called \textit{self-reciprocal} if $\tilde{f}(z) = f(z)$.
	
	Furthermore, for any $f \in \mathbb{F}_2[z]_{<n}$, we define its \textit{reciprocal polynomial with respect to degree $n-1$} as
	\begin{equation}
		f^*(z) = z^{n-1} f(z^{-1}).
	\end{equation}
	
	The operator $f \mapsto f^*$ is an $\mathbb{F}_2$-linear involution on the vector space $\mathbb{F}_2[z]_{<n}$ satisfying $(f^*)^* = f$. 
\end{definition}

The following theorem is the main result of this subsection.

\begin{theorem}\label{thm-bij}
	Let $\alpha$ be a primitive element of $\mathbb{F}_{2^{n}}$, where $n\ge 3$, consider the following two sets: 
	\begin{enumerate}

		\item The set of rational functions \[
		T=
		\left\{
		z^k\frac{\widetilde{u}(z)}{u(z)}
		\,\middle|\,
		u\in \mathbb F_2[z]_{<n},\;
		u(0)=1,\;
		\gcd(u,\widetilde{u})=1,\;
		|k|\le n-1-\deg(u)
		\right\}.
		\]
		
		\item The set  $$\Omega (M)=\{ x\in\mathbb F_{2^n}^*\mid \det (M(x))=0\}.$$

	Then the following evaluation is a bijection between them:
		
		$$z^k\frac{\widetilde{u}(z)}{u(z)} \longmapsto \alpha^k\frac{\widetilde{u}(\alpha)}{u(\alpha)}.$$ Furthermore, each rational function $z^k\displaystyle\frac{\widetilde{u}(z)}{u(z)}$ in $T$ is uniquely determined by the pair $(k,u(z))$. 
	\end{enumerate}
\end{theorem}

We call this bijection a \textit{canonical  representation} of elements in $\Omega (M)$. The canonical representation will not only play an important role in the rank characterization, but also help us to obtain more results about rank dynamics in Section \ref{sec-dynamics}.

We first prove the following theorem, from which Theorem \ref{thm-bij} will follow.

\begin{theorem}
	\label{thm:main}
	Let $n\ge 3$ be an integer. Elements in $$\Omega (M)=\{ x\in\mathbb F_{2^n}^*\mid \det (M(x))=0\}$$ are exactly generated by the rational mapping
	$$
	x = \frac{P^*(\alpha)}{P(\alpha)},
	$$
	where $P(z)$ ranges over all non-zero polynomials in $\mathbb{F}_2[z]_{<n}$.
\end{theorem}

The proof of Theorem \ref{thm:main} splits into two parts. We first deduce that any element in $\Omega (M)$ must admit the representation $x = P^*(\alpha)/P(\alpha)$, and subsequently show that every such element constitutes a root. 

Some auxiliary results will be needed. The following proposition characterizes that, within the kernel of $M(x)$, a vanishing determinant guarantees the presence of a vector whose components follow a Frobenius orbit. We achieve this by fixed-point theorem of semilinear representation presented in the previous section.

\begin{proposition}
	\label{prop:kernel_special_vector}
	If $\det(M(x)) = 0$, then the kernel of $ M(x)$ contains a non-zero vector $c \in \mathbb{F}_{2^n}^n$ of the form 
	\begin{equation}
		c = \left(c_1, c_1^2, c_1^{2^2}, \dots, c_1^{2^{n-1}}\right)^T,
	\end{equation}
	for some $c_1 \in \mathbb{F}_{2^n}^{*}$.
\end{proposition}

\begin{proof}
	By assumption, the $\mathbb{F}_{2^n}$-vector space $V = \ker M(x) \subseteq \mathbb{F}_{2^n}^n$ is non-trivial. For $1 \le i, j \le n$, the matrix entries are given by $M(x)_{i,j} = \frac{1 + x_i x_j}{1 + \alpha_i \alpha_j}$ with $x_i = x^{2^{i-1}}$ and $\alpha_i = \alpha^{2^{i-1}}$. Observe that the matrix entries satisfy the relation
	$$
	(M(x)_{i-1, j-1})^2 = \frac{1+x_{i-1}^2 x_{j-1}^2}{1+\alpha_{i-1}^2 \alpha_{j-1}^2} = \frac{1+x_i x_j}{1+\alpha_i \alpha_j} = M(x)_{i,j},
	$$
	with indices taken modulo $n$, where index $0$ is identified with $n$. 
	
	Define a map $\phi: V \to V$ by $\phi(v)_i = v_{i-1}^2$ for $1 \le i \le n$, i.e., \[
	\phi\!
	\begin{pmatrix}
		v_1\\
		v_2\\
		\vdots\\
		v_n
	\end{pmatrix}
	=
	\begin{pmatrix}
		v_n^2\\
		v_1^2\\
		\vdots\\
		v_{n-1}^2
	\end{pmatrix}.
	\] Indeed, for any $v \in V$, $\phi(v)$ belong to the kernel of $M(x)$, because
	$$
	\begin{aligned}
		(M(x)\phi(v))_i &= \sum_{j=1}^n M(x)_{i,j} \phi(v)_j \\
		&= \sum_{j=1}^n \left( M(x)_{i-1, j-1} \right)^2 v_{j-1}^2 \\
		&= \left( \sum_{j=1}^n M(x)_{i-1, j-1} v_{j-1} \right)^2\\
		&=\left( [M(x)v]_{i-1} \right)^2\\
		&=0.
	\end{aligned}
	$$

	The field extension $\mathbb{F}_{2^n}/\mathbb{F}_2$ is Galois with group $G = \mathrm{Gal}(\mathbb{F}_{2^n}/\mathbb{F}_2) = \langle \sigma \rangle$, where $\sigma(a) = a^2$ is the Frobenius automorphism. Then there is a semilinear representation  $\rho : G\to \Gamma\rm{L}(V)$, defined as $\rho(\sigma^k)= \phi^{k}$. We now verify that this is indeed a semilinear representation:
	\begin{enumerate}
		\item \textit{additivity:}  $\phi(v+w)_i = (v_{i-1}+w_{i-1})^2 = v_{i-1}^2 + w_{i-1}^2 = \phi(v)_i + \phi(w)_i$. Thus $\phi$, and consequently all $\phi^k$, are additive.
		\item\textit{semilinearity:} For any scalar $a \in \mathbb{F}_{2^n}$ and $v \in V$, we have $\phi(av)_i = (av)_{i-1}^2 = a^2 v_{i-1}^2 = \sigma(a)\phi(v)_i$. By induction, $\rho(\sigma^k)$ is semilinear.
		\item \textit{identity and composition:} Evaluating $\phi^n$ on any vector $v \in V$, we get $\phi^n(v)_i = v_{i-n}^{2^n} = v_i^{2^n}$. Since $v_i \in \mathbb{F}_{2^n}$, we have $v_i^{2^n} = v_i$. Therefore, $\rho(\sigma^n)= \rho(\sigma) = id_V$. The composition rule $\rho(\sigma^k) \circ \rho(\sigma^l) = \phi^{k+l \pmod n} = \rho(\sigma^{k+l})$ trivially holds.
	\end{enumerate}
	
	This verification confirms that $G$ acts semilinearly on $V$. Since $V\ne 0$,  Theorem \ref{thm:galois-descent} guarantees that the invariant $\mathbb{F}_2$-subspace $V^G$ is non-trivial. Thus, there exists a non-zero vector $c \in V^G$. 
	
	Being a fixed point implies $\phi(c) = c$. Component-wise, this expands to $c_{i-1}^2 = c_i$ for $2 \le i \le n$, and $c_n^2 = c_1$. By induction, this recurrence directly yields $c_i = c_1^{2^{i-1}}$ for $1 \le i \le n$, which is precisely the desired form for the vector $c$.
\end{proof}

 The following partial fraction identity will be used.

\begin{lemma}
	\label{lem:partial_fraction}
	Let $L(z) \in \mathbb{F}_2[z]$ be a polynomial of degree $n$ with $n$ distinct roots $\alpha_1, \dots, \alpha_n$, and let $L'(z)$ be its formal derivative. For any polynomial $f \in \mathbb{F}_2[z]$ with $\deg(f) < n$, we have
	\begin{equation}
		\frac{f(z)}{L(z)} = \sum_{j=1}^n \frac{f(\alpha_j)/L'(\alpha_j)}{z + \alpha_j}.
	\end{equation}
\end{lemma}

\begin{proof}
	By Lagrange interpolation over the $n$ distinct roots of $L(z)$, we can write
	\[
	f(z) = \sum_{j=1}^n f(\alpha_j) \prod_{k \neq j} \frac{z - \alpha_k}{\alpha_j - \alpha_k}.
	\]
	Since $L(z) = \prod_{i=1}^n (z - \alpha_i)$, its formal derivative evaluated at $\alpha_j$ is exactly $L'(\alpha_j) = \prod_{k \neq j} (\alpha_j - \alpha_k)$. Substituting this into the interpolation formula gives
	\[
	f(z) = \sum_{j=1}^n f(\alpha_j) \frac{L(z)}{(z - \alpha_j)L'(\alpha_j)}.
	\]
	Dividing both sides by $L(z)$ yields the result.
\end{proof}

\begin{proposition}
	\label{prop:root_form}
	If $x$ is a root of $\det(M(x)) = 0$, then $x = P^*(\alpha)/P(\alpha)$ for some non-zero polynomial $P \in \mathbb{F}_2[z]_{<n}$, where $n\ge 3$.
\end{proposition}

\begin{proof}
	Let $L(z) \in \mathbb{F}_2[z]$ be the minimal polynomial of $\alpha$. By Proposition \ref{prop:kernel_special_vector}, assume that $c = \left(c_1, c_1^2, c_1^{2^2}, \dots, c_1^{2^{n-1}}\right)^T$ is a non-zero vector contained in the kernel of $M(x)$. By polynomial interpolation, there exists a unique $p \in \mathbb{F}_2[z]_{<n}$ such that $c_1 = p(\alpha)/L'(\alpha)$. Applying the Frobenius automorphism, we obtain $c_i = p(\alpha_i)/L'(\alpha_i)$ for $1 \le i \le n$. 
	
	Let $d_i = x_i c_i$. Since $x_i = x_1^{2^{i-1}}$ and $c_i = c_1^{2^{i-1}}$, we have $d_i = d_1^{2^{i-1}}$. Thus, there exists a unique $q \in \mathbb{F}_2[z]_{<n}$ such that $d_i = q(\alpha_i)/L'(\alpha_i)$. Because $c_{i}\ne 0$, it follows that $p(\alpha_i) \neq 0$, which implies
	\begin{equation}
		\label{eq:x_i_ratio}
		x_i = \frac{q(\alpha_i)}{p(\alpha_i)}.
	\end{equation}
	
	The condition $M(x)c = 0$ means that for each row $i$,
	\[
	\sum_{j=1}^n \frac{1+x_i x_j}{1+\alpha_i \alpha_j} c_j = 0.
	\]
	Using $x_j c_j = d_j$ and the representations $c_i = p(\alpha_i)/L'(\alpha_i)$ and $d_i = q(\alpha_i)/L'(\alpha_i)$, this becomes
	\begin{equation}
		\label{eq:kernel_ith}
		\alpha_i^{-1} \sum_{j=1}^n \frac{p(\alpha_j)/L'(\alpha_j)}{\alpha_i^{-1} + \alpha_j} + x_i \alpha_i^{-1} \sum_{j=1}^n \frac{q(\alpha_j)/L'(\alpha_j)}{\alpha_i^{-1} + \alpha_j} = 0.
	\end{equation}
	
	By Lemma \ref{lem:partial_fraction}, the sums in (\ref{eq:kernel_ith}) correspond to the partial fraction expansions of $p(z)/L(z)$ and $q(z)/L(z)$ evaluated at $z = \alpha_i^{-1}$. Hence, the equation simplifies to
	\[
	\alpha_i^{-1} \frac{p(\alpha_i^{-1})}{L(\alpha_i^{-1})} + x_i \alpha_i^{-1} \frac{q(\alpha_i^{-1})}{L(\alpha_i^{-1})} = 0.
	\]
	Multiplying by $\alpha_i L(\alpha_i^{-1})$ gives $p(\alpha_i^{-1}) + x_i q(\alpha_i^{-1}) = 0$. Substituting $x_i = q(\alpha_i)/p(\alpha_i)$ from (\ref{eq:x_i_ratio}), we get
	\[
	p(\alpha_i)p(\alpha_i^{-1}) + q(\alpha_i)q(\alpha_i^{-1}) = 0.
	\]
	Multiplying by $\alpha_i^{n-1}$ yields
	\[
	p(\alpha_i)p^*(\alpha_i) + q(\alpha_i)q^*(\alpha_i) = 0.
	\]
	
	Let $H(z) = p(z)p^*(z) + q(z)q^*(z) \in \mathbb{F}_2[z]$, then $H(\alpha_i) = 0$ for $1 \le i \le n$. Observe that
	\begin{align*}
		H(z^{-1}) &= p(z^{-1})p^*(z^{-1}) + q(z^{-1})q^*(z^{-1}) \\
		&= \big(z^{-(n-1)} p^*(z)\big)\big(z^{-(n-1)} p(z)\big) + \big(z^{-(n-1)} q^*(z)\big)\big(z^{-(n-1)} q(z)\big) \\
		&= z^{-(2n-2)} \big( p^*(z)p(z) + q^*(z)q(z) \big) \\
		&= z^{-(2n-2)} H(z).
	\end{align*}
	Then $z^{2n-2}H(z^{-1}) = H(z)$. Evaluating this identity at $z = \alpha_i$ produces
	\begin{equation}
		\alpha_i^{2n-2} H(\alpha_i^{-1}) = H(\alpha_i) = 0.
	\end{equation}
	Since $\alpha_i \neq 0$, this immediately forces $H(\alpha_i^{-1}) = 0$ for all $1 \le i \le n$, providing $n$ additional roots. For $n \ge 3$, the elements $\alpha$ and $\alpha^{-1}$ are not Galois conjugates, meaning the $2n$ roots $\alpha_1, \dots, \alpha_n, \alpha_1^{-1}, \dots, \alpha_n^{-1}$ are all distinct. Since $H(z)$ has degree at most $2n-2$ but possesses $2n$ distinct roots, it must be the zero polynomial, giving
	\begin{equation}
		\label{eq:pp_qq}
		p(z)p^*(z) = q(z)q^*(z).
	\end{equation}
	
	Let $P(z) = p(z) + q^*(z) \in \mathbb{F}_2[z]_{<n}$. Then $P^*(z) = p^*(z) + q(z)$. Using $pp^* = qq^*$, we have
	\begin{align*}
		q(z)P(z) &= q(z)\big(p(z) + q^*(z)\big) \\
		&= q(z)p(z) + q(z)q^*(z) \\
		&= p(z)q(z) + p(z)p^*(z)\\
		&= p(z)\big(q(z) + p^*(z)\big)\\
		&= p(z)P^*(z).
	\end{align*}
	
	Now the form of $x $ can be determined:
	
	\begin{itemize}
		\item \textbf{Case 1: $P(z) \neq 0$.} Since $\deg(P) < n$, $P(\alpha) \neq 0$. Evaluating $qP = pP^*$ at $z = \alpha$ gives
		\[
		x = \frac{q(\alpha)}{p(\alpha)} = \frac{P^*(\alpha)}{P(\alpha)}.
		\]
		
		\item \textbf{Case 2: $P(z) = 0$.} Then $p(z) = q^*(z)$, which implies $p^*(z) = q(z)$. This directly gives
		\[
		x = \frac{q(\alpha)}{p(\alpha)} = \frac{p^*(\alpha)}{p(\alpha)}.
		\]
	\end{itemize}
	In either case, $x$ can be expressed in the required form.
\end{proof}

Having established the necessary form of the zeros, we now prove the converse: any element admitting such a representation necessarily forces the matrix $M(x)$ to be singular.

\begin{proposition}
	\label{prop:sufficiency}
	For any non-zero polynomial $p \in \mathbb{F}_2[z]_{<n}$, the element $x = \displaystyle\frac{p^*(\alpha)}{p(\alpha)}$ is a root of the equation $\det(M(x)) = 0$.
\end{proposition}

\begin{proof}

	Let $x = \displaystyle\frac{p^*(\alpha)}{p(\alpha)}$, then the $(i, j)$-th entry of the matrix $M(x)$ is given by
	\begin{equation}
		M(x)_{i,j} = \frac{1 + x_i x_j}{1 + \alpha_i \alpha_j} = \frac{p(\alpha_i)p(\alpha_j) + p^*(\alpha_i)p^*(\alpha_j)}{p(\alpha_i)p(\alpha_j)(1 + \alpha_i \alpha_j)}.
	\end{equation}
	Using $p^*(z) = z^{n-1} p(z^{-1})$, we have $p^*(\alpha_j) = \alpha_j^{n-1} p(\alpha_j^{-1})$ and $p(\alpha_j) = \alpha_j^{n-1} p^*(\alpha_j^{-1})$. Substituting these into the numerator yields
	\begin{align*}
		p(\alpha_i)p(\alpha_j) + p^*(\alpha_i)p^*(\alpha_j) &= p(\alpha_i) \big(\alpha_j^{n-1} p^*(\alpha_j^{-1})\big) + p^*(\alpha_i) \big(\alpha_j^{n-1} p(\alpha_j^{-1})\big) \\
		&= \alpha_j^{n-1} \big( p(\alpha_i) p^*(\alpha_j^{-1}) + p^*(\alpha_i) p(\alpha_j^{-1}) \big).
	\end{align*}
	Writing the denominator factor as $1 + \alpha_i \alpha_j = \alpha_j (\alpha_i + \alpha_j^{-1})$, the matrix entry becomes
	\begin{equation}
		\label{eq:matrix_entry_bezout}
		M(x)_{i,j} = \frac{1}{p(\alpha_i)} \cdot \frac{p(\alpha_i) p^*(\alpha_j^{-1}) + p^*(\alpha_i) p(\alpha_j^{-1})}{\alpha_i + \alpha_j^{-1}}\cdot \frac{\alpha_j^{n-2}}{p(\alpha_j)}.
	\end{equation}
	
	The middle term in (\ref{eq:matrix_entry_bezout}) is precisely the B{\'e}zoutian polynomial $B(z, y)$ of $p(z)$ and $p^*(z)$, evaluated at $z = \alpha_i$ and $y = \alpha_j^{-1}$. Since both $p$ and $p^*$ have degrees at most $n-1$, their B{\'e}zoutian has a maximum degree of $n-2$ in each variable. Let $C = (C_{r,h})$ be its $(n-1) \times (n-1)$ coefficient matrix, such that
	\begin{equation}
		B(z, y) = \sum_{r=0}^{n-2} \sum_{h=0}^{n-2} C_{r,h} z^r y^h.
	\end{equation}
	Substituting this expansion back into (\ref{eq:matrix_entry_bezout}) gives:
	\begin{equation}
		M(x)_{i,j} = \frac{1}{p(\alpha_i)} \left( \sum_{r=0}^{n-2} \sum_{h=0}^{n-2} \alpha_i^r C_{r,h} (\alpha_j^{-1})^h \right) \frac{\alpha_j^{n-2}}{p(\alpha_j)}.
	\end{equation}
	
	This scalar summation is the exact entry-wise definition of the matrix product $M(x) = \Delta_1 V_1 C V_2^{\top} \Delta_2$, where:
	\begin{itemize}
		\item $\Delta_1 = \mathrm{diag}\big(p(\alpha_1)^{-1}, \dots, p(\alpha_n)^{-1}\big)$,
		\item $\Delta_2 = \mathrm{diag}\big(\alpha_1^{n-2}p(\alpha_1)^{-1}, \dots, \alpha_n^{n-2}p(\alpha_n)^{-1}\big)$,
		\item $V_1 \in \mathbb{F}_{2^n}^{n \times (n-1)}$ with entries $(V_1)_{i,r} = \alpha_i^r$,
		\item $V_2 \in \mathbb{F}_{2^n}^{n \times (n-1)}$ with entries $(V_2)_{j,h} = (\alpha_j^{-1})^h$.
	\end{itemize}
	
	Hence, $\mathrm{rank}(M(x)) = \mathrm{rank}(C) \le n-1$. The $n \times n$ matrix $M(x)$ is therefore singular, concluding the proof.
\end{proof}

Combining Proposition \ref{prop:root_form} and Proposition \ref{prop:sufficiency} together, the proof of Theorem \ref{thm:main} is completed. To derive Theorem \ref{thm-bij} from Theorem \ref{thm:main}, we need several auxiliary results.

\begin{lemma}\label{lem:distinct_roots}
	Let $n \ge 3$, and let $p, q \in \mathbb{F}_2[z]_{<n} \setminus \{0\}$. If $\frac{p^*(\alpha)}{p(\alpha)} = \frac{q^*(\alpha)}{q(\alpha)}$, then their corresponding rational functions are identical in the function field $\mathbb{F}_2(z)$, i.e., $\frac{p^*(z)}{p(z)} = \frac{q^*(z)}{q(z)}$.
\end{lemma}

\begin{proof}
	The condition implies $p^*(\alpha)q(\alpha) + p(\alpha)q^*(\alpha) = 0$. Define the polynomial
	\begin{equation}
		H(z) = p^*(z)q(z) + p(z)q^*(z).
	\end{equation}
	Because $H(z)$ has coefficients in $\mathbb{F}_2$ and $H(\alpha) = 0$, it must also vanish at all $n$ Galois conjugates of $\alpha$. Thus, $H(\alpha_i) = 0$ for $1 \le i \le n$. Using  $p(z^{-1}) = z^{-(n-1)}p^*(z)$ and $q(z^{-1}) = z^{-(n-1)}q^*(z)$, we have:
	\begin{align*}
		H(z^{-1}) &= p^*(z^{-1})q(z^{-1}) + p(z^{-1})q^*(z^{-1}) \\
		&= \big(z^{-(n-1)}p(z)\big)\big(z^{-(n-1)}q^*(z)\big) + \big(z^{-(n-1)}p^*(z)\big)\big(z^{-(n-1)}q(z)\big) \\
		&= z^{-(2n-2)} \big( p(z)q^*(z) + p^*(z)q(z) \big) \\
		&= z^{-(2n-2)} H(z).
	\end{align*}
	Evaluating the identity $z^{2n-2}H(z^{-1}) = H(z)$ at $z = \alpha_i$ gives $\alpha_i^{2n-2}H(\alpha_i^{-1}) = H(\alpha_i) = 0$. This yields $n$ additional roots $H(\alpha_i^{-1}) = 0$. 
	
	For $n \ge 3$, the $2n$ roots $\alpha_1, \dots, \alpha_n, \alpha_1^{-1}, \dots, \alpha_n^{-1}$ are all distinct. Then $H(z)$ must be the zero polynomial. This gives $\frac{p^*(z)}{p(z)} = \frac{q^*(z)}{q(z)}$.

\end{proof}

\begin{lemma}\label{lem:max-selfrec-general}
	Let $0\ne f(z)\in \mathbb F_2[z]$. Then there exist a unique integer $r\ge 0$ and unique
	polynomials $v(z),u(z)\in \mathbb F_2[z]$ such that
	\[
	f(z)=z^r v(z)u(z),
	\]
	where
	\[
	u(0)=v(0)=1,\quad \widetilde{v}(z)=v(z),\quad \gcd(u,\widetilde{u})=1.
	\]
	Moreover, every self-reciprocal divisor of $f(z)$ with nonzero constant term divides $v(z)$.
\end{lemma}

\begin{proof}
	Write $f(z)=z^r h(z)$, where $r\ge 0$ and $h(0)=1$. Factor $h(z)$ into monic irreducible polynomials over $\mathbb F_2$. Since $h(0)=1$, every
	irreducible factor of $h$ has nonzero constant term. For any such irreducible polynomial
	$\phi(z)$, its reciprocal polynomial $\widetilde{\phi}(z)$ is again monic and irreducible.
	Hence the irreducible factors of $h$ are partitioned into two types: self-reciprocal factors
	$\phi=\widetilde{\phi}$, and reciprocal pairs $\{\psi,\widetilde{\psi}\}$ with
	$\psi\ne \widetilde{\psi}$.
	
	Therefore $h(z)$ can be written uniquely in the form
	\[
	h(z)=\prod_{i=1}^s f_i(z)^{a_i}\prod_{j=1}^t g_j(z)^{b_j}\widetilde{g_j}(z)^{c_j},
	\]
	where each $f_i(z)$ is monic irreducible and self-reciprocal, each pair
	$g_j(z),\widetilde{g_j}(z)$ consists of distinct monic irreducible polynomials, and all
	polynomials occurring above are pairwise distinct.
	
	Define
	\[
	v(z):=\prod_{i=1}^s f_i(z)^{a_i}
	\prod_{j=1}^t \bigl(g_j(z)\widetilde{g_j}(z)\bigr)^{\min(b_j,c_j)},
	\]
	and
	\[
	u(z):=\prod_{j=1}^t g_j(z)^{\,b_j-\min(b_j,c_j)}
	\widetilde{g_j}(z)^{\,c_j-\min(b_j,c_j)}.
	\]
	Then
	\[
	f(z)=z^r v(z)u(z).
	\]
	
	By construction, $u(0)=1$. Also, $\widetilde{v}(z)=v(z)$, since each $f_i$ is self-reciprocal and
	each product $g_j\widetilde{g_j}$ is self-reciprocal. Moreover, for each $j$, at most one of
	$g_j$ and $\widetilde{g_j}$ occurs in $u$, so no irreducible factor can divide both $u$ and
	$\widetilde{u}$. Hence
	\[
	\gcd(u,\widetilde{u})=1.
	\]
	
	Let $d(z)$ be a self-reciprocal divisor of $f(z)$ with $d(0)\ne 0$. Then $d(z)\mid h(z)$.
	Since $d(z)=\widetilde{d}(z)$, the exponents of $\psi$ and $\widetilde{\psi}$ in $d(z)$ are equal
	for every reciprocal pair $\{\psi,\widetilde{\psi}\}$. Therefore the exponent of each $f_i$ in
	$d(z)$ is at most $a_i$, and the exponent of each of $g_j$ and $\widetilde{g_j}$ is at most
	$\min(b_j,c_j)$. It follows that
	\[
	d(z)\mid v(z).
	\]
	
	Thus every self-reciprocal divisor of $f(z)$ with nonzero constant term divides $v(z)$. Finally, the integer $r$ is uniquely determined by $f(z)$, and the factorization of $h(z)$ into monic irreducible polynomials is unique. Hence the above construction uniquely determines $v(z)$ and $u(z)$.
	This completes the proof.
\end{proof}

\begin{proposition}\label{lem:S=T}
	In the rational function field $\mathbb F_2(z)$, the following two sets coincide:
	\[
	S=
	\left\{
	\frac{p^*(z)}{p(z)}
	\,\middle|\,
	p\in \mathbb F_2[z]_{<n}\setminus\{0\}
	\right\},
	\]
	and
	\[
	T=
	\left\{
	z^k\frac{\widetilde{u}(z)}{u(z)}
	\,\middle|\,
	u\in \mathbb F_2[z]_{<n},\;
	u(0)=1,\;
	\gcd(u,\widetilde{u})=1,\;
	|k|\le n-1-\deg(u)
	\right\}.
	\]
	Moreover, every element of $T$ is uniquely determined by the pair $(k,u(z))$.
\end{proposition}

\begin{proof}
	We first prove $S\subseteq T$. Take any
	\[
	R(z)=\frac{p^*(z)}{p(z)}\in S,
	\quad
	0\neq p(z)\in \mathbb F_2[z]_{<n}.
	\]
	
	By Lemma \ref{lem:max-selfrec-general}, there exist unique polynomials $v(z),u(z)$ such that $p(z)=z^r v(z)u(z)$, where
	\[
	\widetilde{v}=v,
	\quad
	u(0)=v(0)=1,
	\quad
	\gcd(u,\widetilde{u})=1.
	\]
	Using $\widetilde{v}=v$, we obtain
	\[
	p^*(z)
	=
	z^{n-1}p(z^{-1})
	=
	z^{n-1-r-\deg(v)-\deg(u)}\,v(z)\widetilde{u}(z).
	\]
	Therefore
	\[
	R(z)=\frac{p^*(z)}{p(z)}
	=
	z^k\frac{\widetilde{u}(z)}{u(z)},
	\]
	where $k=n-1-2r-\deg(v)-\deg(u)$. Now set $s:=n-1-\deg(u)$.

	Since $\deg(p)=r+\deg(v)+\deg(u)\le n-1$, we have
	\[
	r+\deg(v)\le s.
	\]
	It follows that
	\[
	k=s-(2r+\deg(v))\le s.
	\]
	Also,
	\[
	2r+\deg(v)\le 2(r+\deg(v))\le 2s,
	\]
	so
	\[
	k=s-(2r+\deg(v))\ge -s.
	\]
	Thus
	\[
	|k|\le s=n-1-\deg(u),
	\]
	and hence $R(z)\in T$. This proves $S\subseteq T$.

	\medskip
	
	Next we prove $T\subseteq S$. Let
	\[
	F(z)=z^k\frac{\widetilde{u}(z)}{u(z)}\in T,
	\]
	where
	\[
	u(0)=1,\quad \gcd(u,\widetilde{u})=1,\quad |k|\le n-1-\deg(u).
	\]
	Set
	\[
	r:=\max\{-k,0\},\quad D:=n-1-\deg(u)-k-2r.
	\]
	Then $r\ge 0$ and $D\ge 0$. Indeed, if $k\ge 0$, then $r=0$ and
	\[
	D=n-1-\deg(u)-k\ge 0,
	\]
	while if $k<0$, then $r=-k$ and
	\[
	D=n-1-\deg(u)+k\ge 0
	\]
	since $|k|\le n-1-\deg(u)$.
	
	Choose any self-reciprocal polynomial $v(z)$ of degree $D$ with $v(0)=1$, and define
	\[
	p(z):=z^r v(z)u(z).
	\]
	Then
	\[
	\deg(p)=r+D+\deg(u)=n-1-k-r\le n-1,
	\]
	so $p(z)\in \mathbb F_2[z]_{<n}\setminus\{0\}$.
	
	Since $\widetilde{v}(z)=v(z)$, we have
	\[
	\frac{p^*(z)}{p(z)}
	=
	z^{n-1-2r-D-\deg(u)}\frac{\widetilde{u}(z)}{u(z)}
	=
	z^k\frac{\widetilde{u}(z)}{u(z)}
	=
	F(z),
	\]
	because $D=n-1-\deg(u)-k-2r$. Thus $F\in S$, and hence $T\subseteq S$. Therefore $S=T$.
	
	\medskip
	
	Finally, we prove the uniqueness of the representation in $T$.
	Assume
	\[
	z^k\frac{\widetilde{u}(z)}{u(z)}
	=
	z^\ell\frac{\widetilde{w}(z)}{w(z)},
	\]
	where both pairs $(k,u)$ and $(\ell,w)$ satisfy the defining conditions of $T$.
	Without loss of generality, assume $k\ge \ell$. Then
	\[
	z^{k-\ell}\widetilde{u}(z)w(z)=\widetilde{w}(z)u(z).
	\]
	
	Evaluating the above identity at $z=0$ gives $k=\ell$. Therefore
	\[
	\widetilde{u}(z)w(z)=\widetilde{w}(z)u(z).
	\]
	Since $\gcd(u,\widetilde{u})=1$, the polynomial $u$ divides $w$. By symmetry, $w$ divides $u$. We conclude that $u=w$.
	
	Thus the representation is unique.
\end{proof}

Now Theorem \ref{thm-bij} is a direct corollary of Proposition \ref{lem:S=T} and Lemma \ref{lem:distinct_roots}.

\subsection{Rank characterization of the associated symmetric matrices}
Based on the established canonical representation of elements in $\Omega(M)$ in Theorem \ref{thm-bij}, using B{\'e}zoutian of polynomials, we prove the following theorem, which is an explicit characterization of $\rank(M(x))$.

\begin{theorem}\label{prop:rank-formula}
	For $n\ge 3$, let $x\in \Omega (M)=\{ x\in\mathbb F_{2^n}^*\mid \det (M(x))=0\}$, and let
	\[
	x=\alpha^{k_0}\frac{\widetilde{u}(\alpha)}{u(\alpha)}
	\]
	be the canonical representation of $x$, where
	\[
	u(0)=1,\quad \gcd(u,\widetilde{u})=1,\quad |k_0|\le n-1-\deg(u).
	\]
	Then
	\[
	\rank(M(x))=\deg(u)+|k_0|.
	\]
\end{theorem}

\begin{proof}
	Choose $0\ne p(z)\in \mathbb F_2[z]_{<n}$ such that
	\[
	x=\frac{p^*(\alpha)}{p(\alpha)}.
	\]
	By Proposition \ref{lem:S=T} and Lemma \ref{lem:max-selfrec-general}, suppose
	\[
	p(z)=z^t v(z)u(z),\quad p^*(z)=z^w v(z)\widetilde{u}(z),
	\]
	where $\widetilde{v}(z)=v(z)$, $\gcd(u,\widetilde{u})=1$, and $k_0=w-t$.
	
	By the factorization in Proposition~\ref{prop:sufficiency},
	\[
	M(x)=\Delta_1V_1CV_2^\top\Delta_2,
	\]
	where $\Delta_1,\Delta_2$ are invertible and $V_1,V_2$ have full column rank. Hence
	\[
	\rank(M(x))=\rank(C),
	\]
	where $C$ is exactly the B{\'e}zoutian
	\[
	B_{n}(p,p*)=\frac{p(z)p^*(y)-p^*(z)p(y)}{z-y}.
	\]
	By Theorem~\ref{lem:bezout-rank}, we have
	\begin{align*}
		\rank(M(x))
		&=\max\{\deg(p),\deg(p^*)\}-\deg(\gcd(p,p^*))\\
		&=\max\{t+\deg(v)+\deg(u),\, w+\deg(v)+\deg(u)\}
		-\deg\bigl(z^{\min(t,w)}v(z)\bigr)\\
		&=\max\{t,w\}+\deg(v)+\deg(u)-\min(t,w)-\deg(v)\\
		&=\deg(u)+|w-t|\\
		&=\deg(u)+|k_0|.
	\end{align*}
\end{proof}

We conclude this section by providing the following example illustrating the canonical representation in Theorem \ref{thm-bij} and the enumeration of singular Gram matrices.

\begin{example}
	For $n=5$, Theorem \ref{thm:rank_distribution} gives
	\[
	\#\{1\le t\le 2^{n}-1\mid\rank(G_tG_t^\top)<n\}=	|T|=2^{n-1}+1=17.
	\]
	
	A direct Magma computation confirms this. More precisely, one checks that
	\[
	T=
	\left\{
	z^k\frac{\widetilde{u}(z)}{u(z)}
	\,\middle|\,
	u\in \mathbb F_2[z]_{<5},\;
	u(0)=1,\;
	\gcd(u,\widetilde{u})=1,\;
	|k|\le 4-\deg(u)
	\right\}
	\]
	contains exactly $17$ distinct rational functions, shown as follows:
	\[
	1,\ z,\ z^2,\ z^3,\ z^4,\ z^{-1},\ z^{-2},\ z^{-3},\ z^{-4},
	\]
	and
	\[
	\frac{z^4+z+1}{z^4+z^3+1},\quad
	\frac{z^3+z+1}{z^4+z^3+z},\quad
	\frac{z^3+z+1}{z^3+z^2+1},\quad
	\frac{z^4+z^2+z}{z^3+z^2+1},
	\]
	\[
	\frac{z^4+z^3+1}{z^4+z+1},\quad
	\frac{z^3+z^2+1}{z^4+z^2+z},\quad
	\frac{z^3+z^2+1}{z^3+z+1},\quad
	\frac{z^4+z^3+z}{z^3+z+1}.
	\]
	
\end{example}

\section{The proof of main results}\label{sec-dynamics}
In this section, we prove our main results stated in the introduction. We first enumerate polynomials over $\mathbb{F}_2$ with non-zero constant term that are coprime to their reciprocal polynomials. This enumeration not only serves as a key ingredient in deriving the rank distribution, but may also be of independent interest.

\begin{proposition}\label{enu of rational functions}
	For each integer $d\ge 0$, we have the following enumeration:
	\[
	\#\{u\in \mathbb F_2[z]\mid \deg(u)=d,\ u(0)=1,\ \gcd(u,\widetilde{u})=1\}
	=\begin{cases}
		1, & d=0,\\[2pt]
		0, & d=1,2,\\[2pt]
		\displaystyle\frac{2^{d-1}-2(-1)^{d}}{3}, & d\ge 3.
	\end{cases}
	\]
\end{proposition}
\begin{proof}
	Let
	\[
	A_d:=\#\{u\in \mathbb F_2[z]\mid \deg(u)=d,\ u(0)=1,\ \gcd(u,\widetilde{u})=1\}.
	\]
	Write their generating function as $U(x):=\displaystyle\sum_{d=0}^{\infty}A_d x^d$, and let $P(x):=\displaystyle\sum_{d=0}^{\infty}p_d x^d$ be the generating function for polynomials $f(z)\in \mathbb F_2[z]$ with $f(0)=1$, where $p_d$
	denotes the number of such polynomials of degree $d$. Since $	p_0=1$, and $p_d=2^{d-1} (d\ge 1)$, we have
	\[
	P(x)=1+\sum_{d=1}^{\infty}2^{d-1}x^d=\frac{1-x}{1-2x}.
	\]
	
	Let $V(x):=\displaystyle\sum_{i=0}^{\infty}v_i x^i$ be the generating function for self-reciprocal polynomials $v(z)$ with $v(0)=1$, where $v_i$ denotes the number of such polynomials of degree $i$. Since such a polynomial is determined by its coefficients of degrees $1,\dots,\lfloor i/2\rfloor$, one has
	\[
	v_i=2^{\lfloor i/2\rfloor},
	\]
	and therefore
	\[
	V(x)=\sum_{i=0}^{\infty}2^{\lfloor i/2\rfloor}x^i=\frac{1+x}{1-2x^2}.
	\]
	
	By Lemma~\ref{lem:max-selfrec-general}, every polynomial $f(z)\in \mathbb F_2[z]$ with $f(0)=1$
	admits a unique factorization
	\[
	f(z)=v(z)u(z),
	\]
	where
	\[
	v(0)=u(0)=1,\quad \widetilde{v}(z)=v(z),\quad \gcd(u,\widetilde{u})=1.
	\]
	Hence $P(x)=V(x)U(x)$, and thus
	
	\begin{equation}\label{eq 2x3}
	U(x)=\frac{P(x)}{V(x)}
	=\frac{(1-x)(1-2x^2)}{(1-2x)(1+x)}.
	\end{equation}
	
Then we have
\begin{equation}
U(x) = 1 + \frac{2x^3}{(1-2x)(1+x)} = 1 + \frac{2x^3}{3}\left(\frac{2}{1-2x} + \frac{1}{1+x}\right),
\end{equation}

expanding each fraction yields

\begin{equation}\label{eq U}
U(x)=1+\frac{2x^3}{3}\left(2\sum_{j=0}^{\infty}(2x)^j+\sum_{j=0}^{\infty}(-x)^j\right)
=1+\frac{2}{3}\sum_{j=0}^{\infty}\left(2^{j+1}+(-1)^{j}\right)x^{j+3}.
\end{equation}
Substituting $d=j+3$ in Eq. (\ref{eq U}) gives

\begin{equation}
U(x)=1+\frac{2}{3}\sum_{d=3}^{\infty}\left(2^{d-2}+(-1)^{d-3}\right)x^{d}=1+\sum_{d=3}^{\infty}\frac{2^{d-1}-2(-1)^{d}}{3}x^d.
\end{equation} 
Then the coefficients $A_0=1$, $A_1=A_2=0$, and for $d\ge 3$,
\[
A_d=\frac{2^{d-1}-2(-1)^{d}}{3}.
\]
 This completes the proof.

\end{proof}

Combining Proposition~\ref{enu of rational functions} with the canonical representation in Theorem~\ref{thm-bij} and the rank characterization in Theorem~\ref{prop:rank-formula}, we are now ready to prove Theorem~\ref{thm:rank_distribution}.

\begin{proof}[The proof of Theorem \ref{thm:rank_distribution}]
By the canonical representation in Theorem~\ref{thm-bij} and the rank characterization in Theorem~\ref{prop:rank-formula}, we have
  \begin{align*}
	N_{k} &:=\#\{x\in \mathbb F_{2^n}^*\mid \rank(M(x))=k\}\\&=\#\left\{
	z^{k_0}\frac{\widetilde u(z)}{u(z)}
	\,\middle|\,
	u\in\mathbb F_2[z]_{<n},\ u(0)=1,\ \gcd(u,\widetilde u)=1,\ \deg(u)+|k_0|=k
	\right\}.
\end{align*} Since every rational function in this set is uniquely determined by the pair $(k_{0},u(z))$, for $0\le k\le n-1$, we have
	\[
	N_k
	=A_k+2\sum_{d=0}^{k-1}A_d,
	\]
where \[
A_d=\#\{u\in \mathbb F_2[z]\mid \deg(u)=d,\ u(0)=1,\ \gcd(u,\widetilde{u})=1\}.
\]
	Write $N(x):=\displaystyle\sum_{k=0}^{\infty}N_k x^k$ and $U(x):=\displaystyle\sum_{d=0}^{\infty}A_d x^d$. Substituting the definition of $N_k$ and interchanging the order of summation, we obtain
	\begin{equation}\label{eq:N-expand}
		N(x)=\sum_{k=0}^{\infty}N_kx^k
		=U(x)+2\sum_{d=0}^{\infty}A_d\sum_{k=d+1}^{\infty}x^k
		=U(x)+\frac{2x}{1-x}\,U(x).
	\end{equation}
	Combining Eq. \eqref {eq:N-expand} with Eq. \eqref{eq 2x3} yields
	\begin{equation}\label{eq:N-closed}
		N(x)=\frac{1+x}{1-x}\,U(x)=\frac{1-2x^2}{1-2x}.
	\end{equation}
	
	Since
	\[
	N(x)=(1-2x^2)\sum_{j=0}^{\infty}2^j x^j
	=1+2x+\sum_{k=2}^{\infty}2^{k-1}x^k,
	\]
	the coefficients are $N_0=1$, $N_1=2$, and for $2\le k\le n-1$, $N_k=2^{k-1}$.
	
	It follows that
	\[
	\#\{x\in \mathbb F_{2^n}^*\mid\rank(M(x))=k\}=
	\begin{cases}
		1, & k=0,\\
		2, & k=1,\\
		2^{k-1}, & 2\le k\le n-1.
	\end{cases}
	\]
	Because
	\[
	\sum_{k=0}^{n-1}N_k
	=1+2+\sum_{k=2}^{n-1}2^{k-1}
	=2^{n-1}+1,
	\]
	we have
	\[
	\#\{x\in \mathbb F_{2^n}^*\mid\rank(M(x))=n\}
	=(2^n-1)-\sum_{k=0}^{n-1}N_k
	=2^{n-1}-2.
	\]
	This proves the theorem.
\end{proof}

We then prove the three dynamical properties, and begin with the persistence phenomenon after full-rank state.

\begin{proof}[The proof of Theorem \ref{thm Persistence}]
	Since
	\[
	r_n(t-1)=\rank(M(\alpha^{t-1}))=n-1,
	\]
	the matrix $M(\alpha^{t-1})$ is singular. By Theorem \ref{thm-bij}, the element $\alpha^{t-1}$ admits a unique canonical representation
	\[
	\alpha^{t-1}=\alpha^{k_0}\frac{\widetilde u(\alpha)}{u(\alpha)},
	\quad
	u(0)=1,\quad \gcd(u,\widetilde u)=1,\quad |k_0|\le n-1-\deg(u).
	\]
	By Theorem \ref{prop:rank-formula}, $n-1=\rank(M(\alpha^{t-1}))=\deg(u)+|k_0|$, hence $|k_0|=n-1-\deg(u)$. If $k_0=-(n-1-\deg(u))$, then
	\[
	|k_0+1|=n-2-\deg(u),
	\]
	and therefore Theorem \ref{prop:rank-formula} gives
	\[
	\rank(M(\alpha^t))
	=\deg(u)+|k_0+1|
	=n-2,
	\]
	contrary to the assumption that
	\[
	r_n(t)=\rank(M(\alpha^t))=n.
	\]
	Thus
	\[
	k_0=n-1-\deg(u).
	\]
	
	It follows that
	\[
	\alpha^{t+1}
	=\alpha^{k_0+2}\frac{\widetilde u(\alpha)}{u(\alpha)}.
	\]
	Since
	\[
	k_0+2=n+1-\deg(u)>n-1-\deg(u),
	\]
	we have
	\[
	\alpha^{t+1}\notin T.
	\]
	Hence $M(\alpha^{t+1})$ is nonsingular, and therefore
	\[
	r_n(t+1)=\rank(M(\alpha^{t+1}))=n.
	\]
\end{proof}

We next show that rank-deficient states are unstable.

\begin{proof}[The proof of Theorem \ref{thm Insta}]
	Set $x=\alpha^t$. Since
	\[
	r_n(t)=\rank(M(x))\le n-1,
	\]
	the matrix $M(x)$ is singular. By Theorem \ref{thm-bij}, $x$ admits a unique canonical representation
	\[
	x=\alpha^{k_0}\frac{\widetilde u(\alpha)}{u(\alpha)},
	\quad
	u(0)=1,\quad \gcd(u,\widetilde u)=1,\quad |k_0|\le n-1-\deg(u).
	\]
	Set $d=\deg(u)$. Then Theorem \ref{prop:rank-formula} gives
	\[
	k=r_n(t)=\rank(M(x))=d+|k_0|.
	\]
	
	If $|k_0+1|\le n-1-d$, then
	\[
	\alpha x=\alpha^{k_0+1}\frac{\widetilde u(\alpha)}{u(\alpha)}
	\]
	is again the canonical representation of $\alpha x$. Hence
	\[
	r_n(t+1)=\rank(M(\alpha x))=d+|k_0+1|.
	\]
	Therefore
	\[
	r_n(t+1)-r_n(t)=|k_0+1|-|k_0|\in\{-1,1\},
	\]
	so
	\[
	r_n(t+1)\in\{k-1,k+1\}.
	\]
	
	Now suppose $|k_0+1|>n-1-d$. Since $|k_0|\le n-1-d$, this forces
	\[
	k_0=n-1-d.
	\]
	Hence
	\[
	k=d+k_0=d+(n-1-d)=n-1.
	\]
	Moreover,
	\[
	\alpha x=\alpha^{k_0+1}\frac{\widetilde u(\alpha)}{u(\alpha)}\notin T.
	\]
	Therefore $M(\alpha x)$ is nonsingular, and thus
	\[
	r_n(t+1)=\rank(M(\alpha x))=n=k+1.
	\]
	
	This proves that
	\[
	r_n(t+1)\in\{k-1,k+1\}.
	\]
\end{proof}

Finally, we characterize the local minima of the rank function and count them.

\begin{proof}[The proof of Theorem \ref{thm Enumeration of local minima}]
	Let
	\[
	\alpha^t=\alpha^{k_0}\frac{\widetilde u(\alpha)}{u(\alpha)}
	\]
	be the canonical representation of $\alpha^t$ given by Theorem \ref{thm-bij}. By Theorem \ref{prop:rank-formula},
	\[
	r_n(t)=\rank(M(\alpha^t))=\deg(u)+|k_0|.
	\]
	
	We first determine when the rank increases by one to the right. If $k_0<n-1-\deg(u)$, then
	\[
	r_n(t+1)=\deg(u)+|k_0+1|.
	\]
	If $k_0=n-1-\deg(u)$, then $\alpha^{t+1}\notin T$, and hence $r_n(t+1)=n=\deg(u)+|k_0|+1$. Therefore,
	\[
	r_n(t+1)=r_n(t)+1
	\iff k_0\ge 0.
	\]
	
	We next determine when the rank increases by one to the left. If $	k_0>-(n-1-\deg(u))$, then
	\[
	r_n(t-1)=\deg(u)+|k_0-1|.
	\]
	If $k_0=-(n-1-\deg(u))$, then $\alpha^{t-1}\notin T$, and hence
	\[
	r_n(t-1)=n=\deg(u)+|k_0|+1.
	\]
	Therefore,
	\[
	r_n(t-1)=r_n(t)+1
	\iff k_0\le 0.
	\]
	
	Thus $t$ is a local minimum if and only if
	\[
	k_0=0.
	\]
	
	Hence the local minima are in one-to-one correspondence with the polynomials
	$u(z)\in\F_2[z]$ satisfying
	\[
	u(0)=1,\quad \gcd(u,\widetilde u)=1,\quad \deg(u)\le n-1.
	\]
	Therefore the number of local minima is
	\[
	\sum_{d=0}^{n-1}A_d.
	\]
	
	Consider the generating function
	\[
	\sum_{m=1}^\infty \left(\sum_{d=0}^{m-1} A_d\right)x^{m-1}.
	\]
	We have
	\[
	\sum_{m=1}^\infty \left(\sum_{d=0}^{m-1} A_d\right)x^{m-1}
	=\sum_{d=0}^\infty A_d\sum_{m=d+1}^\infty x^{m-1}
	=\sum_{d=0}^\infty A_d\frac{x^d}{1-x}
	=\frac{U(x)}{1-x}.
	\]
	Since
	\[
	U(x)=\frac{(1-x)(1-2x^2)}{(1-2x)(1+x)},
	\]
	it follows that
	\[
	\sum_{m=1}^\infty \left(\sum_{d=0}^{m-1} A_d\right)x^{m-1}
	=\frac{1-2x^2}{(1-2x)(1+x)}
	=1+\frac{x}{(1-2x)(1+x)}.
	\]
	Moreover,
	\[
	\frac{1}{(1-2x)(1+x)}
	=\frac{2}{3}\cdot\frac{1}{1-2x}+\frac{1}{3}\cdot\frac{1}{1+x}.
	\]
	Therefore
	\[
	\sum_{d=0}^{n-1} A_d=\frac{2^{n-1}-(-1)^{n-1}}{3}.
	\]
	This completes the proof.
\end{proof}

\section{Hull distribution of punctured cyclic simplex code}\label{sec-LCD}
In this section, we determine the hull distribution of linear codes generated by observability matrices associated with m-sequences, as an application of the established results. In particular, we find that nearly half of them are LCD codes.

Let $\{s_t\}_{t=0}^\infty$ be the binary $m$-sequence considered throughout the paper, where $s_t=\Tr_{\mathbb F_{2^n}/\mathbb F_2}(\lambda\alpha^t)$. Recall that the corresponding observability matrix is
\[
G_t=[v_0,v_1,\dots,v_{t-1}]
=
\begin{pmatrix}
	s_0 & s_1 & \cdots & s_{t-1}\\
	s_1 & s_2 & \cdots & s_t\\
	\vdots & \vdots & \ddots & \vdots\\
	s_{n-1} & s_n & \cdots & s_{n+t-2}
\end{pmatrix}\in \F_2^{n\times t}.
\]

For $n\le t\le 2^{n}-1$, let $\mathcal{C}_t$ be the $[t,n]$ linear code generated by $G_t$. We point out that the code family $\{\mathcal{C}_t\}$ is actually a family of punctured cyclic simplex codes. Indeed, the matrix $G_{2^{n}-1}$ has entries
\[
(G_{2^{n}-1})_{i,j}=\Tr_{\mathbb F_{2^n}/\mathbb F_2}(\lambda \alpha^{i+j-2}).
\]
On the other hand, the binary cyclic simplex code of length $2^n-1$ admits the standard trace representation
\[
\left\{\left(\Tr_{\mathbb F_{2^n}/\mathbb F_2}(\beta),\Tr_{\mathbb F_{2^n}/\mathbb F_2}(\beta\alpha),\dots,\Tr_{\mathbb F_{2^n}/\mathbb F_2}(\beta\alpha^{2^n-2})\right)\mid \beta\in\mathbb F_{2^n}\right\}.
\]
Hence $\mathcal{C}_{2^n-1}$ is exactly the binary cyclic simplex code with parameters $[2^n-1,n,2^{n-1}]$, and therefore each $\mathcal{C}_t$ is a punctured code of this binary cyclic simplex code, obtained by retaining the first $t$ coordinates.

\begin{definition}
	Let $\mathcal{C}$ be a binary linear code. The hull of $\mathcal{C}$ is defined as
	\[
	\textit{Hull}(\mathcal{C}):=\mathcal{C}\cap \mathcal{C}^\perp.
	\]
	The code $\mathcal{C}$ is called a linear complementary dual (LCD) code if
	\[
	\textit{Hull}(\mathcal{C})=\{0\},
	\]
and it is called a self-orthogonal code if $\mathcal{C}\subseteq  \mathcal{C}^\perp$, i.e.,	\[
\textit{Hull}(\mathcal{C})=\mathcal{C}.
\]
\end{definition}

Due to their nice applications, LCD codes and self-orthogonal codes have been extensively studied \cite{LiC 2017,LiC 2025,Li 2023,Li 2024}. The notion of hull is a natural generalization of LCD codes and self-orthogonal codes. It was first introduced by Assmus and Key \cite{Assmus 1990}, aiming to classify finite projective planes. Linear codes with various hull dimensions are attractive since they have important applications in quantum coding theory \cite{Li 2024,Guenda 2018}, and can be used to determine the complexity of some algorithms \cite{Sendrier 2000,Sendrier 2001}. The dimension of hull can be determined by the rank of Gram matrix, and the following proposition is well-known \cite{Ding 2015}. 

\begin{proposition}\label{prop:hull-rank-formula}
	Let $q$ be a prime power, and let $\mathcal{C}\subseteq \F_q^n$ be a linear code with generator matrix $G$. Then,
	
	\[
	\dim(\textit{Hull}(\mathcal{C}))=\dim(\mathcal{C})-\rank(GG^\top).
	\]
\end{proposition}

 By Lemma \ref{lem:decomposition-Gt}, when $n\le t\le 2^{n}-1$, the rank of $G_{t}$ equals $n$, then $\{\mathcal{C}_{t}\mid n\le t\le 2^{n}-1\}$ is a family of $[t,n]$ code. The following characterization of LCD codes is direct by Theorem \ref{thm:main}.
 
 \begin{proposition}\label{prop:LCD-characterization-Ct}
 	Let $\alpha$ be a primitive element of $\mathbb{F}_{{2}^{n}}$, and let $s_t=\Tr_{\mathbb F_{2^n}/\mathbb F_2}(\lambda\alpha^t)$ be an  m-sequence. Let $\mathcal{C}_t$ be the punctured simplex code generated by the observability matrix $G_{t}$ associated with the m-sequence. Let $n\le t\le 2^n-1$, then $\mathcal{C}_t$ is LCD if and only if $$\alpha^t\notin 
 	\left\{
 	\frac{p^*(\alpha)}{p(\alpha)}
 	\,\middle|\,
 	p\in \mathbb F_2[z]_{<n}\setminus\{0\}
 	\right\}.
 	$$
 	
 \end{proposition}
 
  Employing the established results, now we determine hull distribution of the family $\{\mathcal{C}_{t}\mid n\le t\le 2^{n}-1\}$ . 

\begin{proposition}\label{prop:hull-distribution-Ct}
Let $\mathcal{C}_t$ be the simplex code generated by the observability matrix $G_{t}$ associated with an m-sequence of order $n$, where $n\le t\le 2^{n}-1$. For each integer $h$ with $0\le h\le n$, we have
\[
\#\{n\le t\le 2^{n}-1 \mid \dim(\textit{Hull}(\mathcal{C}_t))=h\,\}=
	\begin{cases}
		2^{n-1}-2, & h=0,\\[1mm]
		2^{n-h-1}-1, & 1\le h\le n-2,\\[1mm]
		1, & h=n-1,\\[1mm]
		1, & h=n.
	\end{cases}
\]
\end{proposition}
\begin{proof}
	
Write $B_h:=\#\{n\le t\le 2^{n}-1 \mid \dim(\textit{Hull}(\mathcal{C}_t))=h\,\}$. For $n\le t\le 2^n-1$, by Proposition~\ref{prop:hull-rank-formula}, 
	\[
	\dim(\textit{Hull}(\mathcal{C}_t))=n-\rank(G_tG_t^\top),
	\]
thus  $$B_h=\#\{n\le t\le 2^n-1\mid \rank(G_tG_t^\top)=n-h\}.$$

Recall that in Theorem \ref{thm:rank_distribution}
we have obtained the following rank distribution for the Gram matrices:
\begin{equation}\label{eq rank-dis-hull}
\#\{\,t\in\{1,\dots,2^n-1\}\mid \rank(G_tG_t^\top)=k\,\}
=
\begin{cases}
	1, & k=0,\\[1mm]
	2, & k=1,\\[1mm]
	2^{k-1}, & 2\le k\le n-1,\\[1mm]
	2^{n-1}-2, & k=n.
\end{cases}
\end{equation}

	For $1\le t\le n-1$, the canonical representation of $\alpha^t$ is $\alpha^t\cdot\widetilde{1}/1$, Theorem~\ref{prop:rank-formula} gives
	\[
	\rank(G_tG_t^\top)=\rank(M(\alpha^t))=\deg(1)+|t|=t.
	\]

	 Subtracting the contributions of $t=1,\dots,n-1$ from Eq. \eqref{eq rank-dis-hull} , we obtain the hull distribution:
\[
B_h=
\begin{cases}
	2^{n-1}-2, & h=0,\\[1mm]
	2^{n-h-1}-1, & 1\le h\le n-2,\\[1mm]
	1, & h=n-1,\\[1mm]
	1, & h=n.
\end{cases}
\]
	This completes the proof.
\end{proof}

In particular, in the family $\{\mathcal{C}_{t}\mid n\le t\le 2^{n}-1\}$, there are exactly $2^{n-1}-2$ LCD codes, and only one self-orthogonal code, namely the cyclic simplex code $\mathcal{C}_{2^{n}-1}$.

\section{Concluding Remarks}\label{sec-conclu}

In this paper, we studied the family of Gram matrices $\{G_tG_t^\top\mid 1\le t\le 2^n-1\}$ arising from binary $m$-sequences. We characterized all singular matrices in this family by a special class of rational functions over $\mathbb F_2$, determined their ranks explicitly, and obtained the exact rank distribution. We also investigated the rank dynamics along the natural cyclic ordering and counted the local minima of the rank function. As an application, we characterized when the associated codes $C_t$ are LCD and determined the hull distribution in the family $\{C_t\mid  n\le t\le 2^n-1\}$.

Some open problems remain. A natural $q$-ary analogue can be defined similarly. Let $\{s_j\}_{j=0}^{\infty}$ be the $q$-ary m-sequence $s_j=\Tr_{\mathbb F_{q^n}/\mathbb F_q}(\lambda\alpha^j)$, where $\alpha$ is a primitive element in $\mathbb{F}_{q^{n}}$. Let $G_t^{(q)}$ be the corresponding observability matrix, where $1\le t\le \frac{q^n-1}{q-1}$. Then $G_{\frac{q^n-1}{q-1}}^{(q)}$ exactly generates the $\left[\frac{q^n-1}{q-1},\,n,\,q^{n-1}\right]$  simplex code:
\[
\left\{\left(\Tr_{\mathbb F_{q^n}/\mathbb F_q}(\beta),\Tr_{\mathbb F_{q^n}/\mathbb F_q}(\beta\alpha),\dots,\Tr_{\mathbb F_{q^n}/\mathbb F_q}\!\left(\beta\alpha^{\frac{q^n-1}{q-1}-1}\right)\right)\mid \beta\in\mathbb F_{q^n}\right\}.
\]

Thus $\{G_t^{(q)}\mid n\le t\le\frac{q^n-1}{q-1} \}$ gives rise to a punctured family of $q$-ary simplex codes. By Proposition~\ref{prop:hull-rank-formula}, the hull distribution of this family is completely determined by the rank distribution of the corresponding Gram matrices. Our extensive Magma computations suggest the following open problems.

\begin{open}
Let $q=3$. Determine the rank distribution of the family
	\[
	\left\{\,G_t^{(3)}(G_t^{(3)})^\top \;\middle|\; 1\le t\le \frac{3^n-1}{2}\right\}
	\]
	for ternary $m$-sequences. In particular, prove that for every integer $n\ge 2$, one has
	\[
	\#\left\{\,1\le t\le \frac{3^n-1}{2} \;\middle|\; \rank\!\bigl(G_t^{(3)}(G_t^{(3)})^\top\bigr)=k \right\}
	=
	\begin{cases}
		1, & k=0,\\[1ex]
		\dfrac{3^k-(-1)^k}{2}, & 1\le k\le n-1,\\[2ex]
		\dfrac{3^n-(-1)^n}{4}-1, & k=n.
	\end{cases}
	\]
\end{open}

\begin{open}
Let $q=5$. Prove that for every integer $n\ge 2$, one has
	\[
	\#\left\{\,1\le t\le \frac{5^n-1}{4} \;\middle|\; \rank\!\bigl(G_t^{(5)}(G_t^{(5)})^\top\bigr)=k \right\}
	=
	\begin{cases}
		0, & k=0,\\[1ex]
		\dfrac{5^k-(-1)^k}{3}, & 1\le k\le n-1,\\[2ex]
		\dfrac{5^n-(-1)^n}{6}, & k=n.
	\end{cases}
	\]
\end{open}

It would be of interest to investigate the rank distribution over general finite fields and to determine whether any a common pattern can be found.

\end{document}